\documentclass[11pt]{article}

\usepackage[letterpaper,margin=1in]{geometry}
\usepackage{authblk}
\usepackage{times}
\usepackage{graphicx} 
\usepackage{amsmath,amssymb, amsthm}
\usepackage{hyperref}
\usepackage{cleveref}
\usepackage{xcolor}
\usepackage{dsfont}
\usepackage{mathtools}
\usepackage{tikz}
\usepackage{soul}
\usepackage{booktabs}
\usepackage{subcaption}
\usepackage{siunitx}
\usepackage{centernot}
\usepackage{thm-restate}
\usepackage{natbib}

\usepackage{amsmath,amsfonts,bm}

\def\eqref#1{equation~\ref{#1}}

\def\appref#1{Appendix~\ref{#1}}

\def\1{\bm{1}}

\def\vzero{{\bm{0}}}

\def\vmu{{\bm{\mu}}}

\def\vtheta{{\bm{\theta}}}
\def\vepsilon{{\bm{\varepsilon}}}

\def\vlambda{{\bm{\lambda}}}

\def\vt{{\bm{t}}}
\def\vu{{\bm{u}}}

\def\vx{{\bm{x}}}
\def\vy{{\bm{y}}}
\def\vz{{\bm{z}}}

\def\evtheta{{\theta}}

\def\evt{{t}}

\def\mA{{\bm{A}}}

\def\mI{{\bm{I}}}

\DeclareMathAlphabet{\mathsfit}{\encodingdefault}{\sfdefault}{m}{sl}
\SetMathAlphabet{\mathsfit}{bold}{\encodingdefault}{\sfdefault}{bx}{n}

\def\gC{{\mathcal{C}}}

\def\gF{{\mathcal{F}}}

\def\gI{{\mathcal{I}}}

\def\gL{{\mathcal{L}}}
\def\gM{{\mathcal{M}}}
\def\gN{{\mathcal{N}}}

\def\gP{{\mathcal{P}}}

\def\gS{{\mathcal{S}}}
\def\gT{{\mathcal{T}}}
\def\gU{{\mathcal{U}}}
\def\gV{{\mathcal{V}}}

\def\gX{{\mathcal{X}}}

\newcommand{\R}{\mathbb{R}}

\DeclarePairedDelimiter{\norm}{\lVert}{\rVert}

\newtheorem{definition}{Definition}

\newtheorem{example}{Example}
\newtheorem{lemma}{Lemma}

\title{Axiom Satisfiability of Linear Rewards in Alignment}

\author[1]{Soumya Nasipuri%
    \thanks{Corresponding author. All authors contributed equally to this work.}}
\author[2]{Sayak Ray Chowdhury}
\author[1]{Sanjukta Roy}

\affil[1]{{Indian Statistical Institute, India}
\texttt{\{soumyanasipuri\_r,sanjukta\}@isical.ac.in}}

\affil[2]{Indian Institute of Technology Kanpur, India \\ \texttt{sayakrc@iitk.ac.in}}

\begin{document}

\maketitle

\begin{abstract}
Learning from human preference data is the dominant route to aligning language models with human values. In the linear social choice model, where rewards are linear in a fixed feature representation of prompt-response pairs, \citet{geAxiomsAIAlignment2024} show that fitting such a reward by minimizing any non-decreasing convex loss, including Bradley-Terry-Luce (BTL), fails Pareto Optimality (PO) and Pairwise Majority Consistency (PMC). Moreover, no rule that reads only the majority relation can satisfy PO once the output is required to be linearly induced. We ask what it costs to enforce these axioms anyway. To this end, we relax the linear model to allow per-candidate additive slack. We compute the relaxed linear reward with the smallest total slack that satisfies the axioms with a margin $\eta$, the minimum required difference between two reward values. Our solution satisfies the axioms under no assumptions about the voters or how comparisons were collected. We bound the optimal total slack by $O(1)$ when $\eta$ is at most $O(\frac{1}{m^2})$ for $m$ candidates. Furthermore, we exhibit an instance that forces this bound, concluding that the rate is tight up to constants. In practice, the number of candidates far exceeds the feature dimension, and only a linear reward can be evaluated on unseen responses. We therefore introduce a new method that charges the linear part for each comparison it gets wrong. It simultaneously minimizes the total slack and the number of violations, with a parameter $\lambda$ that trades off between them. We show that the total slack is monotone but saturating in $\lambda$: raising it reduces the violations of the deployed linear reward and increases the slack, yet the slack stays below $(\eta + \Delta\sqrt{d})\lfloor m^2/4\rfloor$, where $\Delta$ and $d$ are the diameter and dimension of the features, respectively. Experiments on both synthetic and real-life preference data with per-annotator labels corroborate our theoretical result and show that the linear reward output by our method beats linear BTL.
\end{abstract} 

\section{Introduction}
AI Alignment is the problem of ensuring that language models act in ways consistent with human intentions, preferences, and normative constraints \citep{halpernAIAlignmentSocialSurvey2026}. 
Reward models trained from human feedback are the interface through which
aggregated human preferences enter a language model
\citep{christiano2017,stiennon2020,ouyang2022,bai2022}. Such models are typically fit by maximum likelihood under a Bradley-Terry model
\citep{bradley1952} on pairwise comparisons, which produces a single scalar
reward from the judgments of many annotators.
When those judgments disagree, the reward model is performing an act of
aggregation. How can we identify the best way to aggregate preferences? Social choice theory \citep{arrow2012,brandt2016} has developed a rich axiomatic framework for addressing this question. The axiomatic approach characterizes aggregation rules through precisely defined, normatively meaningful properties. Thus, it is natural to ask which axioms a reward model satisfies. A recent line of
work argues that this question should be central to alignment
\citep{conitzer2024,dai2024}. 

Given a prompt, an annotator ranks a set of responses. We view each annotator as a voter and each prompt–response pair as a candidate.  Comparisons between every two candidates induce binary preference labels, which are then used to learn a reward model. A loss-based rule assigns rewards by minimizing a measure that penalizes disagreements with voters over candidate pairs.\cite{geAxiomsAIAlignment2024} 
and, subsequently, \citet{hollender2026enforcing} show, surprisingly, an optimally trained linear reward violates a unanimous preference, known as Pareto optimality (PO), i.e, even when everyone selects candidate $a$ over candidate $b$, the learned reward of $b$ is higher. It also violates the preference of a majority, known as pairwise majority consistency (PMC), i.e., even when a majority of voters selects candidate $a$ over candidate $b$, the reward of $b$ is higher. Needless to say, a final ranking representing the human feedback should agree with the unanimous and majority preferences. Thus, in search of an aggregation method, 
\cite{geAxiomsAIAlignment2024} studied the linear reward models. They gave a linear rule based on \cite{Copeland1951}, a well-known rule in social choice theory, that satisfies both of the above constraints. However, a desired ranking need not always be representable by a simple linear function, or even by a polynomial function, of the candidates' feature vectors.

This work asks what it costs to impose the axioms anyway. We relax the linear model of \citet{geAxiomsAIAlignment2024} by a per-candidate \emph{slack}, ask for the smallest slack consistent with the
axioms, and study it as a measure of how far the preference profile lies from what the feature geometry can express.

\subsection{The Linear Social Choice Model~\citep{geAxiomsAIAlignment2024}}
Let $\gC$ be the set of $m$ distinct prompt-response pairs, referred to as
\emph{candidates}, and $\gV$ be the set of $n$ human annotators, the \emph{voters}. Each candidate $c\in \gC$ is associated with a distinct feature vector
$\vx_c \in \mathbb{R}^{d}$. Let $\gX$ denote the set of feature vectors of the candidates $\gC$. We abuse notation and use $\gC$ to denote the set of candidates along with their feature vectors. A parameter vector
$\vtheta \in \mathbb{R}^{d}$ induces a linear reward function 
$r_\vtheta : \gC \mapsto \R$ defined by $r_\vtheta(c)= \langle \vtheta,\vx_c \rangle$ for each candidate $c \in \gC$. We are interested in how any such $r_\vtheta$ ranks the candidates. 
%
For a ranking $\succ_{\vtheta}$ over the candidates, we say that $\vtheta$
\emph{induces} $\succ_{\vtheta}$ if
\begin{align}\label{eq:relation}
  a \succ_{\vtheta} b
\quad \Longleftrightarrow \quad
r_\vtheta(a)\ge r_\vtheta(b).  
\end{align}
We say $\pi = (\succ_1, \succ_2, \ldots, \succ_n)$ a \emph{preference profile}, where ranking $\succ_i$ of $\gC$ is given by voter $i \in \gV$. A \emph{linear aggregation rule} takes as input a profile $\pi$ and outputs a parameter vector $\vtheta^*$. The goal is to design a rule that finds a $\vtheta^*$ whose induced ranking respects some desirable properties (the \emph{axioms}) with respect to $\pi$. In alignment, the model size (e.g., the feature dimension $d$) is much smaller than the prompt/response space (e.g., the number of candidates $m$). In this case ($d < m$), some ranking $\succ$ might not be realized by any $\vtheta$, and hence, the linear aggregation rule might not output a $\vtheta^*$ satisfying the axioms. The example below shows this.
%

\begin{example}\label{ex:challenge}
Consider three candidates $a,b,c$ in $\mathbb{R}$ with $x_a=0$, $x_b=1$, and $x_c=2$. All three lie on a line, so $r_\theta$ is monotone along the line. If $\theta > 0$, then $\theta$ induces the ordering $c \succ b \succ a$, and if $\theta < 0$, then it induces the ordering $a \succ b \succ c$. Out of the $3!$ rankings of $\{a, b, c\}$, only two are
realizable by a linear reward. The remaining four cannot be induced by any $\theta$.

\end{example}
A linear reward orders the
candidates by projecting onto a single direction. So the rankings available to it are constrained by the features $\{\vx_c\}_{c \in \gC}$, and a profile aggregated from voters whose disagreements carry no obligation to respect that constraint. Two improvements follow immediately: enrich the feature representation until the desired ranking becomes realizable, or relax the reward model. We take the second approach by allowing small perturbations to the linear rewards.  

\subsection{Our Contributions}\label{sec:relaxation} 
We relax the linear reward by a per-candidate slack $\varepsilon_c$, yielding
\begin{equation}\label{eq:relaxed-reward}
 r_{\vtheta,\vepsilon}(c) :=\langle \vtheta,\vx_{c} \rangle +\varepsilon_c , \qquad \textit{ for } c \in \gC .
\end{equation}
In the linear model, the slack vector $\vepsilon = \bm{0}$; allowing $\vepsilon \ne \bm{0}$
admits a deviation from linearity, and $\|\vepsilon\|_1$ measures the total absolute slack. By allowing slack in linear rewards, we ensure that all rankings of $\gC$ are realized. e.g., in \Cref{ex:challenge}, the ranking $b \succ c \succ a$ 
isn't realized by linear rewards. But, with the slacks, we can realize it with, say, $\theta=1$ and $\vepsilon = (0, 1.001, 0)$.

We want each slack to be bounded and as small as possible, since otherwise, one can realize any ranking of $\gC$ by taking $\vtheta = \bm{0}$ and only using $\vepsilon$. Our goal is to measure \emph{how much} slack a given profile forces to satisfy the axioms, i.e., how far the axioms push the reward away from the linear class. Formally, we ask for the slack vector $\vepsilon \in \mathbb{R}^m$ of
least $\ell_1$-norm under which Pareto Optimality (PO) and Pairwise Majority Consistency (PMC) hold. We make the following contributions:

(i) We formulate minimizing total slack as a linear program (LP), explicitly enforcing PO and PMC as margin constraints. Condorcet cycles are handled by removing a feedback arc set that is disjoint from the unanimous pairs, thereby preserving as many majority constraints as the cycles allow. Apart from PO and PMC, the reward $r_{\vtheta,\vepsilon}$ returned by the program satisfies Condorcet consistency, and hence, majority consistency, under no assumptions on the profile and feature vectors. 

(ii) In practice, the number of candidates $m$ is much larger than the parameter dimension $d$. So, only the linear model $\vtheta \in \mathbb{R}^d$ is deployable, not the slack $\vepsilon \in \mathbb{R}^m$. However, $\vtheta$ returned from the LP formulated above may exhibit both PO and Majority violations.
We introduce a second program that jointly minimizes the total number of violations by $r_\vtheta$ and the total slack, with a parameter $\lambda$ controlling their relative importance.
The trade-off it governs is monotone: increasing $\lambda$ reduces the violations of $r_\vtheta$ and increases $\norm{\vepsilon}_1$. However, $\norm{\vepsilon}_1$ saturates at $(\eta + \Delta\sqrt{d})\lfloor m^2/4\rfloor$ for $\lambda \ge \Delta\sqrt{d}/\eta$, where $\Delta$ is the $\ell_2$ diameter of the candidate features and $\eta$ is the margin, i.e., the minimum required difference between two reward values. We construct instances, where the total slack $\norm{\vepsilon}_1$ is at least $\tfrac{\eta}{2}\lfloor m^2/4\rfloor$, and show that the upper bound is tight on a class of (infinitely many) instances.




(iii)We characterize when linear rewards can model the preferences. We show that a positive slack does not always imply that the axioms are not realizable with a linear reward $r_\vtheta$: for $\eta$ not larger than the highest attainable margin by $r_\vtheta$, a zero-slack certifies realizability. Otherwise, the certificate is the non-existence of a zero net feature-weighted flow on the majority graph.
 

(iv) Finally, we compare PO and Majority violations of our deployed linear reward against those of a linear model obtained by minimizing the Bradley-Terry loss on both synthetically generated instances and real-life \emph{Machine Deliberation} dataset. We observe a significant reduction in total axiom violations compared to linear BTL, e.g., on the Deliberation dataset, we see a $16.12 \%$ reduction.



\subsection{Related Work}\label{sec:related} \citet{geAxiomsAIAlignment2024} introduces the linear social choice model and shows that every loss-based rule with a nondecreasing weakly convex or strictly convex loss (e.g., BTL) fails both PO and PMC\footnote{The proof contains an incorrect value of $\vtheta$, which can be corrected to establish the statement.}, and that the failure of PO extends to every $C1$ rule. They answer this with Leximax Copeland subject to PO (LCPO), a rule that fills the ranking positions in Copeland order, admitting a candidate when a linear reward can place it there consistently with
Pareto dominance, and prove it satisfies PO and PMC. \cite{hollender2026enforcing} shows that enlarging the reward class does not help, since loss-based rules over bounded-degree polynomial rewards still fail PO and PMC. PO is satisfied under balanced data coverage where comparison directions are uniformly distributed over a sphere, though PMC still fails. Both works assume that the axioms can be met: LCPO requires a PO-respecting feasible ranking to exist, which
\citet{geAxiomsAIAlignment2024} secures by assuming every voter's ranking is itself induced by a parameter vector, and \citet{hollender2026enforcing} requires control over how comparisons are sampled. Instead, we take the features and comparisons as given and ask how far the profile lies from what the geometry can express: the slack is zero exactly when a linear reward satisfies the axioms unaided, and, otherwise, it measures the shortfall.
%
%
\citet{xiao2025theoretical} weight the BTL loss by majority indicators, implement the Copeland rule, thereby satisfying PO and PMC. Their reward class, however, is unrestricted.
\citet{dong2026representation} study this trade-off between representation and rationalizability of axioms.
\citet{procacciaCloneRobustAIAlignment2025} shows that weighted MLE is robust to cloning candidates, preserving other desirable properties of MLE, while \citet{siththaranjanDistributionalPreferenceLearning2024} showed that MLE using regularized BTL loss outputs Borda \citep{borda1784memoire} ranking. 
\cite{halpernAIAlignmentSocialSurvey2026} provides a comprehensive survey of recent related works.

\section{Social Choice Axioms}\label{sec:prelim}
In social choice theory, an axiom constrains an aggregation rule, which maps a profile of individual rankings to a collective ranking. In this work, we examine several axioms from social choice theory to evaluate the reasonableness of our aggregation rule, which is done using a reward function. The rule takes the profile $\pi$ and returns a reward $r$, and the collective ranking is the one induced by ranking candidates in decreasing reward. We therefore state each axiom as a condition on $r$. 


As in \citet{geAxiomsAIAlignment2024}, we primarily consider $C1$ rules, which only need to know, for each pair, which candidate the majority prefers. These axioms constrain $r$ only through pairwise comparisons on which the voters agree to some degree, so we record two sets of pairs:
the \emph{unanimous} pairs $\gU$, on which every voter ranks $a$ above $b$, and the \emph{strict majority} pairs $\gM$, on which more than half do.
To this end, let $\pi = (\sigma_1, \sigma_2, \ldots, \sigma_n)$ be a preference profile, where $\sigma_i$ is a ranking of $\gC$ given by voter $i \in \gV$.
For $a, b \in \gC$, write $w_{a \succ b}(\pi) = \tfrac{1}{n}\big|\{i \in \gV : a \succ_{i} b \}\big|$ for the fraction of voters
ranking $a$ above $b$, so that $w_{a\succ b} + w_{b \succ a} = 1$.
Define
\begin{align}
\gU &:= \{(a,b) \in \gC \times \gC : w_{a \succ b}(\pi) = 1\} && \text{(unanimous pairs)}, \label{eq:U}\\
\gM &:= \{(a,b) \in \gC \times \gC : w_{a \succ b}(\pi) > \tfrac12\} && \text{(strict majority pairs)}. \label{eq:M}
\end{align}
The set $\mathcal{M}$ defines a (weighted) directed graph $G_{\mathcal{M}}$, called the weighted majority graph, with vertices $\gV$, edges $\gM$ with weights $w_{a \succ b}(\pi)$. For an odd number of voters, the majority graph is complete, i.e., for all $a, b \in \gC$, either $(a,b) \in \gM$ or $(b,a) \in \gM$ (called a tournament~\citep{laslierTournamentSolutionsMajority1997}). 
\begin{definition}[PO and PMC]\label{def:axioms}
On a preference profile $\pi$, a reward function $r: \gC \to \mathbb{R}$ satisfies (i) Pareto Optimality (PO) if $r(a) > r(b)$ for every $(a,b) \in \gU$, and (ii) Pairwise Majority Consistency (PMC) if the graph $G_{\gM}$ has no cycle and $r(a) > r(b)$ for every $(a,b) \in \gM$.
\end{definition}
Observe that if there is a cycle, no ranking can satisfy PMC for all pairs, and the axiom is vacuous. Moreover, a PMC ranking for a profile need not exist, but when it does, it is unique.
A cycle in the graph $G_{\gM}$ is called a \emph{Condorcet cycle}.
A weaker version of PMC is \emph{Condorcet consistency}: if a candidate $c^*$ is preferred by a majority of voters over every other candidate, then $c^*$ must be ranked first. Equivalently, if $w_{c^* \succ c} > \frac{1}{2}$ for each $c \in \gC \setminus \{c^*\}$, then  $r(c^*) > r(c)$. The candidate $c^*$ is called the Condorcet winner. A further relaxed version is majority consistency. A rank aggregation rule satisfies \emph{majority consistency (MC)} if, when a candidate $a$ is ranked first by a majority of voters in the input profile, $a$ is ranked first in the output ranking. 


Neither PO nor PMC is satisfied with the rules ordinarily used. Every loss-based rule with a convex loss fails, including BTL. Moreover, no $C1$ rule can satisfy PO once the output is required to be induced by a linear reward. 
Leximax Copeland subject to PO (LCPO) works when the preference profile is linearly realizable, and  LCPO satisfies both PO and PMC~\citep{geAxiomsAIAlignment2024}.
A profile $\pi$ is \emph{linearly realizable} if there exists a $\vtheta \in \mathbb{R}^d$ s.t $r_\vtheta(a)> r_{\vtheta}(b)$ for all $(a,b) \in \gM$ when the PMC ranking exist, otherwise the condition holds for each $(a,b) \in \gU$.
A linear reward may not satisfy these axioms. As Example~\ref{ex:challenge} shows, the rankings a linear function can induce are limited by the geometry of the feature vectors, and the axioms may demand one that lies outside that set -- for any candidate set with $m > d+1$, most rankings are unavailable. We therefore relax the model with per-candidate slack~\eqref{eq:relaxed-reward} and ask for the smallest slack under which the axioms hold. 

\section{A Solution via Linear Programming}

The axioms of Section~\ref{sec:prelim} translate directly into margin
constraints. For a margin $\eta> 0$, PO asks for 
$\vtheta^\top(\vx_a - \vx_b) + (\vepsilon_a - \vepsilon_b) \ge \eta$ on
every unanimous pair, and PMC asks for the same on every strict majority pair.
The latter cannot always be imposed: when $\gM$ contains a Condorcet cycle, no
reward orders all of it consistently~\citep{liu2026statistical}, and the natural remedy is to drop the
majority constraints entirely in that case, enforcing them only when $\gM$ is
acyclic. Writing $\mathds{1}_{\text{PMC}}$ for the indicator of acyclicity, this
gives the constraint family $\gU \cup \mathds{1}_{\text{PMC}}\gM$. 

The relation $\succ$ of
\eqref{eq:relation} is defined with a weak inequality, under which
$\vtheta = \bm{0}$ induces every ranking. Thus, we ask for strict separation at a fixed margin $\eta > 0$, which also
keeps the feasible set closed. Moreover, without a bound on $\vtheta$, one could scale
$\vtheta \mapsto \alpha\vtheta$ with $\alpha \to \infty$ to meet every margin
while driving $\vepsilon$ to zero, so we constrain $\vtheta$ in the box $\|\vtheta\|_\infty \le 1$, which fixes the scale and makes $\vepsilon$ a meaningful measure of deviation from linearity.\footnote{\citet{hollender2026enforcing} observe that the Bradley--Terry loss already imposes an implicit budget on $\|\vtheta\|$, as unbounded scaling drives the loss to infinity.}

Minimizing the total slack $\|\vepsilon\|_1$ subject to these constraints is a linear program:
\begin{equation}\label{eq:lp}
\begin{aligned}
(\mathsf{LP1}) &\qquad \min_{\vtheta,
\vepsilon,\vt}\quad  \sum_{a \in \gC} \evt_a \nonumber\\
\text{subject to}\quad & -t_a \le \varepsilon_a \le t_a, \quad t_a \ge 0 && \forall a \in \gC,\nonumber\\
& \langle\vtheta, \vx_a - \vx_b\rangle + (\varepsilon_a - \varepsilon_b) \ge \eta && \forall (a,b) \in \gU \cup \mathds{1}_{\text{PMC}}\gM, \qquad\text{(PO+PMC)} \nonumber\\
& -1 \le \evtheta_j \le 1 && \forall j \in \lbrace 1,\ldots, d\rbrace.  \nonumber 
\end{aligned}
\end{equation}
Only $(\vtheta, \vepsilon)$ is of interest in what follows: the auxiliary variables satisfy $t_a = |\varepsilon_a|$ at any optimum, so they carry no information beyond
$\vepsilon$ and serve only to write $\|\vepsilon\|_1$ linearly.
\begin{restatable}[$\mathsf{LP1}$ satisfies PO and PMC]{lemma}{lemFirstLPAxioms}\label{lemma:primal_criteria}
    Every feasible solution of $\mathsf{LP1}$  induces a reward function $r_{\vtheta,\vepsilon}$ (\eqref{eq:relaxed-reward}) that satisfies Pareto optimality and pairwise majority consistency, but need not satisfy Condorcet consistency.
\end{restatable}
The proof is given in \appref{appdx:proofs}. Intuitively, when $\mathcal{M}$ has a cycle, the constraint family collapses
to $\mathcal{U}$, and a candidate ranked first by a majority of voters need not be preferred to every other candidate unanimously. A single cycle anywhere in $\mathcal{M}$ therefore costs every majority constraint. To recover them by deleting some pairwise constraints, we need the following definition.
\begin{definition}\label{def:fas}
Let $G = (V, E)$ be a directed graph. A \emph{feedback arc set (FAS)} of $G$
is a set $\mathcal{F} \subseteq E$ such that $(V, E\setminus\mathcal{F})$
acyclic; it is \emph{minimal} if no proper subset of it is a feedback arc set. 
\end{definition}

The (weighted) majority graph $G_{\gM}$ defined by $\mathcal{M}$ is a directed graph on $\gC$. We take $\mathcal{F}$
to be a feedback arc set of $G_{\mathcal{M}}$ such that
$\mathcal{F} \cap \mathcal{U} = \emptyset$, and set
\begin{align}\label{eq:constraint-family}
  \mathcal{P} \;:=\; \mathcal{U} \cup (\mathcal{M} \setminus \mathcal{F}).
\end{align}
The omission of the pairs in $\mathcal{F}$ ensures a valid program, where the cardinality of $\mathcal{F}$ affects how many of the majority constraints survive. One can compute a $5$-approximation of the minimum FAS in time $O(|\gC|^2)$ in (weighted) tournaments by simply ordering the vertices in non-decreasing order of their (weighted) in-degrees \citep{coppersmithOrderingWeightedNumber2010}. With this, we obtain a new linear program $\mathsf{LP2}$ by replacing the (PO+PMC) constraint of $\mathsf{LP1}$ with
\begin{align}\label{eq:ppmc}
\mathsf{(LP2)}:\quad\quad \langle\vtheta,\vx_a-\vx_b\rangle+(\varepsilon_a-\varepsilon_b)
\geq\eta
\qquad\forall(a,b)\in\mathcal{P}.
\end{align}
PO and PMC are the axioms we set out to enforce, and the program enforces them
by construction, each being written as an explicit family of margin constraints.
We show in \Cref{thm:primal_lp} that Condorcet consistency, consequently majority consistency, follows without further constraints, as a consequence of the structure of $\gP$. The proof is in \appref{appdx:proofs}. 
\begin{restatable}[$\mathsf{LP2}$ satisfy all axioms]{theorem}{thmSecondLPAxioms}\label{thm:primal_lp}
    Every feasible solution of  $\mathsf{LP2}$ induces a reward function $r_{\vtheta,\vepsilon}$ (\eqref{eq:relaxed-reward}) that satisfies Pareto optimality, pairwise majority consistency, Condorcet consistency, and majority consistency.
\end{restatable}
\emph{Proof Sketch.}
When $G_{\gM}$ is acyclic, then $\gF = \emptyset$, so $\gP = \gM \supseteq \gU$ and \eqref{eq:ppmc} enforce PO and PMC outright. When $G_{\gM}$ is cyclic, only the pairs of $\gM \setminus \gF$ are constrained, and the majority constraint holds on those by construction; PO holds as $\gF \cap \gU = \emptyset$, and thus $\gU$ in $\gP$. For Condorcet consistency, suppose a candidate $a$ is ranked before candidate $b$ by a majority of voters, for every candidate $b \in \gC$. Then $a$ is a source vertex in the majority graph defined by $\gM$. Consequently, no arc out of $a$ lies on a directed cycle, and $a$ will be ranked first with the highest reward. Therefore, Condorcet consistency holds and, by definition, satisfies majority consistency as well. 

\paragraph{Comparison with LCPO.} LCPO~\citep{geAxiomsAIAlignment2024} achieves the same axioms by filling the $m$ positions of the output ranking one at a time. For each position, they test candidates in the order of their Copeland score until a feasible one is found. LCPO solves
$O(m^2)$ linear programs in the worst case.  In contrast, $\mathsf{LP2}$ is a single linear program whose constraints are indexed by $\gP$, which is supported on the comparisons actually observed. Computing $\gP$ costs $O(m^2)$ time by ordering candidates on in-degree and deleting the resulting back arcs, after which a single linear program in $d + 2m$ variables and $O(|\gP|)$ constraints is solved -- as opposed to $O(m^2)$ programs for LCPO, each in $d$ variables and carrying all of $\gU$.

Moreover, LCPO needs a parameter vector that separates every unanimous pair, since each of its feasibility tests yields all Pareto-dominances. When one is absent, it cannot fill even the first position and returns nothing. Such a vector is obtained by assuming that each voter's ranking is itself induced by a $\vtheta'$. This assumption is not testable without eliciting complete rankings, and it restricts each voter to $m^{O(d)}$ linearly induced orderings out of $m!$ possible ones. In contrast, $\mathsf{LP2}$ is feasible on every instance, returning a reward that satisfies the axioms together with an optimal slack $\vepsilon^*$. So, we lose nothing in instances where LCPO applies, but we have a measurement available when it does not.

\section{Moving Closer to Linear Rewards}
\label{sec:prac} 
In practice, we may have $m \gg d$ since the number of prompt-response pairs can be large. Thus, it might not be practical to release both $\vtheta$ and $\vepsilon$ for model evaluation. We therefore ask whether we can get a ``good'' ranking of the candidates based on the linear reward using only using the parameter vector $\vtheta$ obtained from $\mathsf{LP2}$?
Let the ranking of candidates induced by the linear reward $r_\vtheta$ be denoted by $\succ_{\vtheta}$. Then, we want to find a $\vtheta$ such that the number of unanimous and strict majority pairs (as defined in \Cref{eq:U,eq:M}) that are inverted in $\succ_{\vtheta}$ is minimized. Each inverted pair is either a PO violation or violates a majority pair that is not PO (majority violation), indicating a violated constraint in the LP. Since we have no assumption on the preference profile being linearly realizable, we may have some violations for every $\vtheta$ (as in~\Cref{ex:challenge}). Formally, we ask:
\begin{quote}
Given a profile $\pi$ and features $\gX$, can we obtain a $\vtheta \in \mathbb{R}^d$ efficiently such that the number of PO and Majority violations is minimized?
\end{quote}
Given a preference profile $\pi$, let $\succ^*$ be the optimal one that minimizes the number of violations with respect to the PO and majority pairs in $\gP$. A standard measure of distance between two rankings is the \emph{Kendall-Tau (KT) distance}\footnote{KT distance is well known in social choice literature, and the voting rule that produces a ranking minimizing the KT distance from $\succ_i$ for each voter $i \in \mathcal{V}$ is called the Kemeny rule \citep{kemeny1959mathematics}.}. Given two permutations $\succ^*$ and $\succ$ of $[m]$, the KT distance $\mathrm{KT}(\succ^*,\succ)$ is defined as the number of pairs $(i,j)$ such that $i\succ^*j$ and $j\succ i$ for $1\leq i<j\leq m$. Since $\gP$ can be a partial ordering, we define \emph{KT distance of $\succ$ from $\gP$} to be the number of pairs $(i,j) \in \gP$ such that $j\succ i$.
 Let $(\vtheta^*,\vepsilon^*)$ denote an optimal solution of $\mathsf{LP2}$. We first show that $\succ_{\vtheta^*}$ does not minimize the KT distance from a ranking that orders the pairs according to $\gP$.


\begin{example}\label{ex:LP2notminKT}
Let $\gC=\{c_1,c_2,c_3,c_4\}$, with $1$-D features 
$x_{c_1}=-3, x_{c_2}=-2, x_{c_3}=-1$, and $ x_{c_4}=60$.
Suppose that there is only one voter in the profile $\pi$ with preference $\succ =
(c_3\succ c_2\succ c_4\succ c_1)$.
Then, the set of unanimous pairs 
$\gU=\{(c_3,c_2),(c_3,c_4),(c_3,c_1),
(c_2,c_4),(c_2,c_1),(c_4,c_1)\}$.
The majority graph constructed from the preference $\succ$ contains no cycles.
So $\mathcal{F}=\emptyset$ and $\gP=\gU$ in $\mathsf{LP2}$.
We run $\mathsf{LP2}$ with $\eta=1$ and obtain $\| \vepsilon^*\|_1=3.032$ and $\theta^*=-0.01612$. 
Now, for any $\vtheta < 0$, the ranking according to $r_{\vtheta}$ is $\succ'=c_1\succ c_2\succ c_3\succ c_4$. Therefore, $\mathsf{KT}(\succ',\succ)=4$ for $\theta^*$.
Whereas, if we had $\vtheta > 0$, then the ranking according to $r_{\vtheta}$ would be $\widehat{\succ}=c_4\succ c_3\succ c_2\succ c_1$ and
$\mathsf{KT}(\widehat{\succ},\succ)=2$. 
So, while $\mathsf{LP2}$ does minimize the $L_1$ norm of $\vepsilon$, it does not
give us the $\vtheta$ which minimizes the KT distance
from the desired profile.
\end{example}



\subsection{Kendall-Tau Minimizing Model}
We want to find a $\vtheta$ that minimizes the number of PO and Majority violations. In order to do this, we minimize the KT distance from $\gP$.
In the objective function, we can first minimize the KT distance, then minimize the $\ell_1$-norm of $\vepsilon$. Towards this, we introduce a variable $u_{ab}$ for each pair $(a,b) \in \gP$. If $\langle\vtheta, \vx_a-\vx_b\rangle < \eta$, then $u_{ab} =1$ and it is zero otherwise. That is, the variables $u_{ab}$ for $(a,b)\in \gP$ count the number of violated constraints when we only consider $\vtheta$ to compute the reward. Recall that there are no inversions when the reward is computed using $(\vtheta,\vepsilon)$. Thus, $\sum_{(a,b)\in \gP} u_{ab}$ is the total number of violations by $r_{\vtheta}$. Equivalently, it is the total number of inversions when we obtain a ranking $\succ_{\vtheta}$ using $r_{\vtheta}$. Thus, this term is the least KT distance between $\succ_{\vtheta}$ from $\gP$. Therefore, we design the following integer LP.
\begin{align}
\mathsf{(ILP)} &\qquad \min_{\vtheta,
\vepsilon,\vt, \vu}\quad  \sum_{a \in \gC} \evt_a + \lambda \sum_{(a,b) \in \gP}\!\! u_{ab}\nonumber\\
\text{subject to}\quad & -t_a \le \varepsilon_a \le t_a, \quad t_a \ge 0 &&\forall a \in \gC,\nonumber\\
& \langle\vtheta, \vx_a - \vx_b\rangle + (\varepsilon_a - \varepsilon_b) \ge \eta && \forall (a,b) \in \gP, \nonumber\\
& u_{ab}L \ge \eta - \langle\vtheta, \vx_a - \vx_b\rangle, \quad u_{ab} \in\{0,1\}
   && \forall (a,b) \in \gP, \label{eq:countViolations}\\
& -1 \le \evtheta_j \le 1 && \forall j \in \lbrace 1,\ldots, d\rbrace \nonumber. 
\end{align}
We set $L$ to be a large constant with $L \geq \eta+\Delta\sqrt{d}$. When $\lambda$ is a large enough, e.g., $O(m^2)$, we realize our desired goal of first minimizing the KT distance and then minimizing the $\ell_1$-norm of $\vepsilon$. Only $\vtheta$ is of our interest in this ILP, and the auxiliary variables $u_{ab}$ for each $(a,b)\in \gP$ indicate the violations and are not useful for constructing the reward function. 
As the $\mathsf{ILP}$ cannot be efficiently solved, we relax the constraints $u_{ab} \in\{0,1\}$ by $u_{ab} \ge 0$ to obtain a linear program $\mathsf{(LP3)}$. In the LP we do not need the large constant $L$, so we remove it from \Cref{eq:countViolations}.  

We begin by proving the correctness of $\mathsf{LP3}$. Using \Cref{thm:primal_lp} and the fact that the feasible solutions of $\mathsf{LP3}$ are also feasible solutions of $\mathsf{LP2}$, we get the next lemma. All the complete proofs from this section are in Appendix~\ref{appdx:section4}.
\begin{restatable}[$\mathsf{LP3}$ satisfy all axioms]{lemma}{lemThirdLPAxioms}\label{thm:primal_lp3}
Every feasible solution of  $\mathsf{LP3}$ induces a reward function $r_{\vtheta,\vepsilon}$ (\eqref{eq:relaxed-reward}) that satisfies all three of Pareto optimality, pairwise majority consistency, Condorcet Consistency, and majority consistency.
\end{restatable}
Now, we upper bound the total slack of $\mathsf{LP3}$ and provide a corresponding lower bound.
\begin{restatable}[Bounds on total slack]{theorem}{thmThirdLPSlack}\label{thm:upper-bound-LP3}
Let $\Delta := \max_{(a,b)\in \gP}\|\vx_a - \vx_b\|_2$. For every $\lambda \ge 0$,
the optimal slack $\vepsilon^*_\lambda$ of $\mathsf{LP3}$ satisfies
\[
\|\vepsilon^*_\lambda\|_1 \le
\min\big\{\eta\lfloor\tfrac{m^2}{4}\rfloor + \lambda\,\eta\,|\gP|\, ,\,
 \big(\eta + \Delta\sqrt{d}\big)\lfloor\tfrac{m^2}{4}\rfloor \big\}.
\]
For every $m,d,\eta$, there exist features in $\mathbb{R}^d$ and a profile, for which 
$\|\vepsilon^*_\lambda\|_1 \ge \frac{\eta}{2}\lfloor\tfrac{m^2}{4}\rfloor$. 
\end{restatable}
The first upper bound $\eta\bigl\lfloor m^2/4\bigr\rfloor + \lambda\eta|\gP|$ is a feature-free ($\vtheta = \vzero$) baseline and serves as a loose bound that always holds. Setting the margin $\eta < \min \lbrace \tfrac{1}{m^2}, \tfrac{1}{\lambda |\gP|} \rbrace$, we can drive $\|\vepsilon^*_\lambda\|_1 $ to $O(1)$. 
Moreover, in all lower bound instances, $\Delta \leq \eta/(2\sqrt{d})$. On these instances, infinitely many of them, $ \|\vepsilon^*_\lambda\|_1 \le \tfrac{3\eta}{2} \lfloor \tfrac{m^2}{4}\rfloor$, which matches the lower bound within a factor of $3$. 

Setting $\lambda=0$ reduces $\mathsf{LP3}$ to $\mathsf{LP2}$, and our bound gives $\|\vepsilon^*\|_1 \le \eta\lfloor m^2/4\rfloor$, where $\vepsilon^*$ is an optimal solution of $\mathsf{LP2}$. Thus, the total slack of $\mathsf{LP2}$ is less than that of $\mathsf{LP3}$ in the worst-case, but this comes at a cost of more PO and Majority violations, considering only the linear part of the reward (see~\Cref{fig:mallows-dispersion_11_50_25}). As we increase $\lambda$ until it reaches $\Delta \sqrt{d}/\eta$, the total number of constraint violations decreases, and (upper bound on) total slack increases linearly with rate $\eta |\gP|$. Beyond this point, total slack remains capped at $ \big(\eta + \Delta\sqrt{d}\big)\lfloor\tfrac{m^2}{4}\rfloor$ (see~\Cref{fig:lambda_variability_fig}).

Next, we show the dependence of the objective function of the $\mathsf{ILP}$ on $\lambda$. Using an upper bound on the total slack for the $\mathsf{ILP}$ (similar to \Cref{thm:upper-bound-LP3}), we can compute the required value of $\lambda$.
\begin{restatable} [Increasing $\lambda$ decreases violated pairs]{lemma}{lemILPlambda}\label{lem:ILPlambda}
    Let $\vtheta$ denote a solution of $\mathsf{ILP}$ for a given $\lambda$ and $\succ_{\vtheta}$ denote the ranking induced by the reward $r_\vtheta$. Then, the KT distance of $\succ_\vtheta$ from $\gP$ is non-increasing in $\lambda$. After a certain cut-off $\lambda^*$, the number of violated pairs remains the same.  
\end{restatable}
Due to \Cref{eq:countViolations} in the $\mathsf{ILP}$, variable $u_{ab}$ acts as an indicator of a violation by the linear reward for each pair $(a,b)$ in $\gP$.
We next show that when $\lambda$ satisfies the cut-off shown in \Cref{lem:ILPlambda}, i.e., $\lambda > \big(\eta + \Delta\sqrt{d}\big)\lfloor\tfrac{m^2}{4}\rfloor$, then the second term in the objective function is minimized.  
\begin{restatable}[$\mathsf{ILP}$ minimizes KT distance]{lemma}{lemminKTdist}\label{lem:minKTdist}
  Let $\vtheta$ be a solution of $\mathsf{(ILP)}$ and  $\succ_{\vtheta}$ denote the ranking induced by the reward $r_\vtheta$. Then, there exists a $\lambda$ such that the KT distance of $\succ_\vtheta$ from $\gP$ is minimized. 
\end{restatable}

\begin{example}
In continuation of \Cref{ex:LP2notminKT}, we see that for $\mathsf{LP3}$, if we focus on the Kendall-Tau distance minimization in the objective by setting $\lambda=30$, we obtain $\vtheta=0.0158$ 
and $\| \vepsilon \|_1=4.95$. Therefore, the ranking based on $r_\vtheta$ is $\widehat{\sigma}=c_4\succ c_3\succ c_2\succ c_1$  with KT distance $2$ as opposed to $4$ as in $\mathsf{LP2}$. Thus, we get an optimal $\vtheta$ at the cost of increasing the $\ell_1$ norm of $\vepsilon$.
\end{example}

\subsection{Certificate of Linear Realizability}

Given a set of candidates $\gC$ and their feature vectors $\gX=\lbrace\vx_a\rbrace_{a \in \gC}$, recall that a profile $\pi$ on $\gC$ is \emph{linearly realizable} if there exists a $\vtheta \in \mathbb{R}^d$ such that $\langle\vtheta,\vx_a-\vx_b\rangle > 0$ for all $(a,b) \in \gP$. Define the \emph{best attainable margin} by a normalized linear reward
\begin{align}\label{eq:kappa}
 \kappa(\pi, \gX) \;:=\; \max_{\|\vtheta\|_\infty \le 1}\ \min_{(a,b)\in\gP}\
\big\langle\vtheta,\, \vx_a - \vx_b\big\rangle.   
\end{align}
\begin{restatable}[Realizability $\centernot\implies$ zero-slack always]{proposition}{proprealizability}
\label{prop:realizablity}
If $\pi$ is linearly realizable, then $\kappa(\pi,\gX) > 0$ and
$\vepsilon^*_\lambda = \bm{0}$
if and only if $\eta \le \kappa(\pi,\gX)$ for every $\lambda \ge 0$. 
\end{restatable}
A non-zero optimal slack $\vepsilon^*_\lambda$, therefore, doesn't certify non-realizability. We give a characterization of linear realizability based on the feature vectors. Collect the feature differences over constrained pairs into the rows of $\mA:= \big[(\vx_a - \vx_b)^{\!\top}\big]_{(a,b)\in\gP}$. A profile is linearly realizable precisely when some $\vtheta$ has $\mA\vtheta > \bm{0}$ component-wise. Gordan's theorem~\citep{vandenberghe2004} states that this fails if and only if the rows of $\mA$ admit a
non-trivial nonnegative combination summing to zero, i.e., for some $\vmu \in \mathbb{R}^{|\gP|},\vmu \gneq \bm{0}$,  we have that $\mA^{\!\top}\vmu=\bm{0}$. 

Let $\mu_{ab}$ denotes the component of $\vmu$ for the pair $(a,b)\in \gP$. The weights $\mu_{ab} \ge 0$ have a natural reading as a \emph{flow} in the majority graph with edges $\gP$. 
It is the amount of flow along the edge
$(a,b)$. Regrouping the sum by candidates, we get 
\begin{equation*}
\mA^{\!\top}\vmu=\sum_{(a,b)\in \gP}\mu_{ab}(\vx_a - \vx_b) = \sum_{a\in\gC}\nu_a(\vmu)\,\vx_a,\, \text{ where }\, \nu_a(\vmu) := \sum_{b:(a,b)\in\gP}\mu_{ab} - \sum_{b:(b,a)\in\gP}\mu_{ba}.
\end{equation*}
The \emph{net flow} at candidate $a$ is $\nu_a(\vmu)$. The total flow $\sum_{a\in \gC}\nu_a(\vmu) = 0$ since each pair $(a,b) \in \gP$ contributes $+\mu_{ab}$ to $\nu_a(\vmu)$. A non-zero $\vmu$ with
$\sum_a \nu_a(\vmu)\vx_a = \bm{0}$ is called \emph{feature-balanced}. Therefore, either  $\mA\vtheta > \bm{0}$ component-wise, or $\vmu$ is feature balanced.

\begin{restatable}[Certificate of realizability]{proposition}{propgordan}\label{prop:gordan}
A profile $\pi$ is linearly realizable if and only if no non-zero feature balanced flow exists, i.e., $\nexists\, \vmu \gneq \vzero$, with $\sum_{a\in \gC}\nu_a(\vmu)\vx_a = \vzero$.
\end{restatable}
We now record when a feature-balanced flow exists. To this end, we say the feature vectors $\gX$ are in \emph{general position} if every subset $\gT$ of at most $d+1$ of them is affinely independent, that is, there exist coefficients $(\zeta_a)_{a \in \gT}$, not all zero,
with $\sum_{a \in \gT} \zeta_a\, \vx_a = \vzero$ and
$\sum_{a \in \gT} \zeta_a = 0$. 
\begin{restatable}{corollary}{corraffineindep}
\label{prop:affine-indep}
If $\gX \!\subset\! \mathbb{R}^d$ is in general position and $m \!\le\! d+1$, then every profile is linearly realizable. 
\end{restatable}
\Cref{fig:mallows-dispersion_3_6_5} shows linear realizability when $m = d + 1$, as there are no violations of $\mathsf{LP2}$ and $\mathsf{LP3}$.



\begin{figure}[t]
    \centering
    \includegraphics[width=\linewidth]{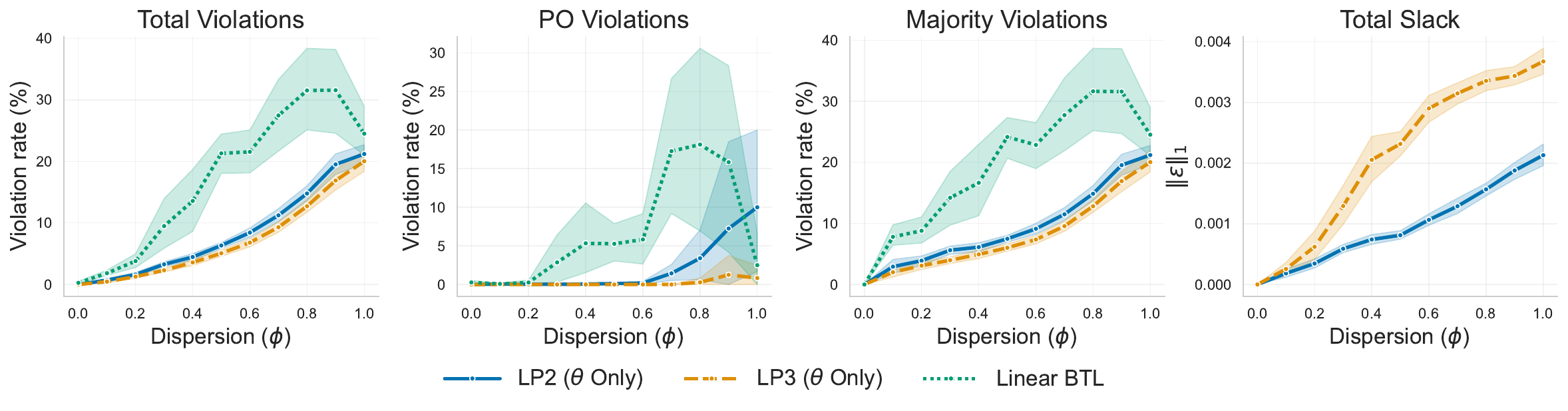}
    \caption{
        Axiom violations and candidate-level slack as a function of normalized Mallows dispersion $\phi$.
        Curves show mean violation rates across trials, with shaded regions denoting 95\% confidence intervals. Panels report (a) Total violations, (b) PO violations, (c) Majority violations,
        and (d) Total $\|\vepsilon\|_1$
        for $\mathsf{LP2}$ and $\mathsf{LP3}$. Experiments use $11$ voters, $50$ candidates, $25$-dimensional features and 20 seeds.
    }
    \label{fig:mallows-dispersion_11_50_25}
\end{figure}
\begin{figure}[t]
    \centering
    \includegraphics[width=\linewidth]{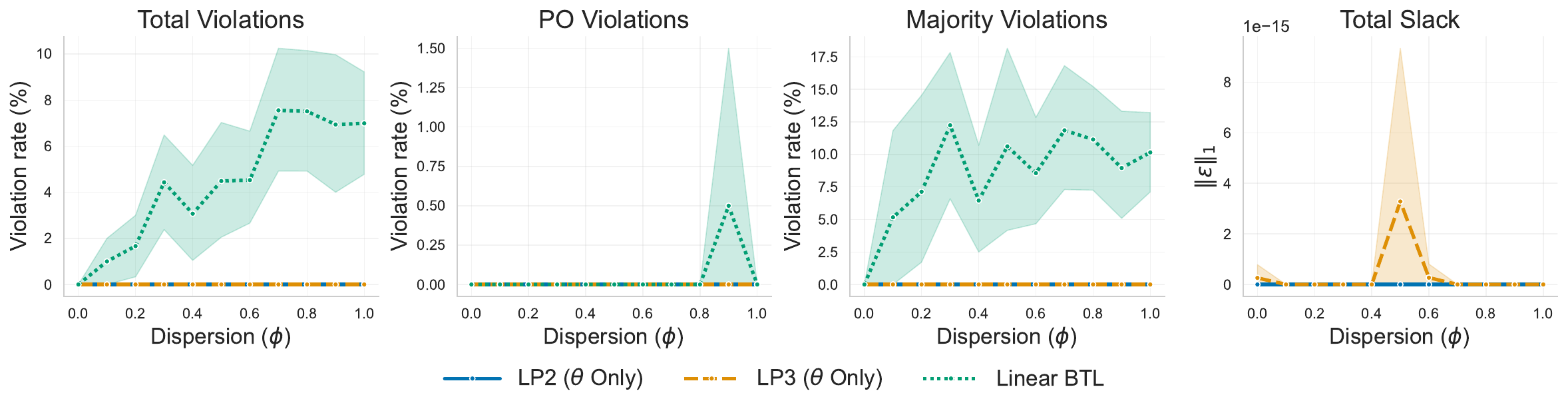}
    \caption{
        Axiom violations and candidate-level slack as a function of normalized Mallows dispersion $\phi$.
        Curves show mean violation rates across trials, with shaded regions denoting 95\% confidence intervals. Panels report (a) Total violations, (b) PO violations, (c) Majority violations,
        and (d) The $\|\vepsilon\|_1$
        for $\mathsf{LP2}$ and $\mathsf{LP3}$.
        Experiments use $3$ voters, $6$ candidates, $5$-dimensional feature embeddings and $20$ seeds.
    }
    \label{fig:mallows-dispersion_3_6_5}
\end{figure}

\section{Experiments} \label{sec:experiments}
We evaluate $\mathsf{LP2}$, $\mathsf{LP3}$, and Linear BTL
\citep{geAxiomsAIAlignment2024} on controlled synthetic and real-world preference datasets. Violations are evaluated using the learned linear reward parameter $\vtheta$, without applying the candidate-specific slack $\vepsilon$, thereby isolating the agreement between the linear representation and the observed preference structure.
We report \emph{PO Violations} and \emph{Majority Violations} as percentages of the corresponding unanimous and majority relations in $\gP$, and \emph{Total Violations} as the percentage of all relations in $\gP$ reversed by $\vtheta$. The source code and experimental implementations are publicly available at \url{https://github.com/Itachi1999/ai_alignment_linear_model_extension}.


\paragraph{Synthetic Data.} We study the effect of increasing preference heterogeneity under controlled feature geometry. Candidate features are sampled from $\gN(\vzero,\mI)$. Preferences are generated according to the normalized Mallows model~\citep{boehmer2021NormalizedMallows} with
$\phi\in\{0,0.1,\ldots,1.0\}$, where larger $\phi$ corresponds to greater dispersion from the reference ranking. The reference ranking is generated using a linear utility of the candidate feature vector and a $\hat{\vtheta}\sim \gN(\vzero,\mI)$. $\mathsf{LP2}$ and $\mathsf{LP3}$ use margin $\eta=10^{-5}$, while $\lambda=10$ for $\mathsf{LP3}$. Figure~\ref{fig:mallows-dispersion_11_50_25} reports the violations (PO, Majority, total), as well as the $\ell_1$ norm of the slack for $\mathsf{LP2}$ and $\mathsf{LP3}$. 


The results show a clear separation between the methods as preference dispersion increases. $\mathsf{LP3}$ consistently exhibits the lowest observed violations (total, PO, majority) across the dispersion range. $\mathsf{LP2}$ incurs the smallest  $\|\vepsilon\|_1$, reflecting its direct minimization of the total slacks. Linear BTL exhibits substantially larger violation rates, with the difference becoming more pronounced as the preference profile becomes more dispersed. These results also illustrate the intended tradeoff introduced by $\mathsf{LP3}$ between candidate-wise slack and pairwise inversions, as slack is more in $\mathsf{LP3}$ than $\mathsf{LP2}$ while all violation rates are less for $\mathsf{LP3}$.

\paragraph{Effect of $\lambda$.} We perform a study on the effects of the parameter $\lambda$ in $\mathsf{LP3}$. Preference profile, features, and reference ranking are generated as before. We keep the dispersion parameter $\phi = 0.5$ and vary $\lambda \in \{0, 20, 40, \dots, 200 \}$. \Cref{fig:lambda_variability_fig} reports the total violations and slack of $\mathsf{LP2}$ and $\mathsf{LP3}$, as well as the margin shortfall $\sum_{(a,b)\in \gP}u_{ab}$ of $\mathsf{LP3}$. We see that the margin shortfall decreases monotonically as $\lambda$ increases. The total violations and the total slack also show desired behaviors as we increase the $\lambda$, but they are not exactly monotonous as we relax $\mathsf{ILP}$ to obtain $\mathsf{LP3}$. 



\begin{figure}[t]
    \centering
    \includegraphics[width=\linewidth]{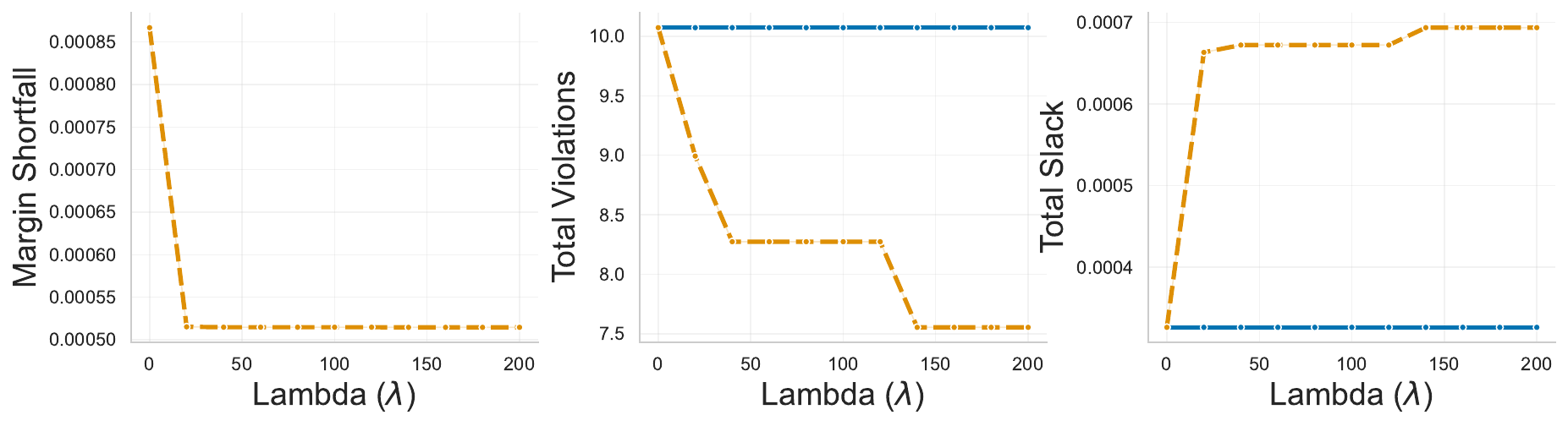}
    \caption{
        Total violations, total slack, and margin shortfall as a function of $\lambda$. The blue line represents $\mathsf{LP2}$, and the orange $\mathsf{LP3}$. Panels report (a) Margin shortfall ($\sum_{(a,b)\in \gP} u_{ab}$), (b) Total violations (\%), (c) The $\|\vepsilon\|_1$. Experiment uses $3$ voters, $25$ candidates, and $10$-dimensional feature.
    }
    \label{fig:lambda_variability_fig}
\end{figure}

\begin{table}[t]
    \caption{Machine Deliberation dataset. $K=500$ questions, $m = 2000, d=768,\eta = 10^{-5}$.}
    \label{tab:real-500-main}
    \centering
    \small
    \begin{tabular}{lccccc}
        \toprule
        Method
        & Total Viol. (\%)
        & PO Viol. (\%)
        & Majority Viol. (\%)
        & $\|\vepsilon\|_1$
        & Time (s) \\
        \midrule
        $\mathsf{LP2}$
        & $1.59 \pm 1.16$
        & $0.65 \pm 0.63$
        & $1.91 \pm 1.37$
        & $\textbf{0.00207} \pm \textbf{0.001}$
        & $\mathbf{308.96}$ \\

        $\mathsf{LP3}$
        & $\mathbf{1.39 \pm 0.98}$
        & $\mathbf{0.44 \pm 0.35}$
        & $\mathbf{1.72 \pm 1.21}$
        & $0.00249 \pm 0.00123$
        & $722.08$ \\

        Linear BTL
        & $17.51 \pm 0.56$
        & $2.24 \pm 0.53$
        & $22.71 \pm 0.62$
        & --
        & $2979.11$ \\
        \bottomrule
    \end{tabular}
\end{table}

\paragraph{Real-World Preference Data.}
We evaluate $\mathsf{LP2}$, $\mathsf{LP3}$, and Linear BTL on \emph{Habermas Machine Deliberation} dataset~\citep{tesslerPreflibDataset2024}. It contains 2,710 separate elections, each consisting of five voters ranking four LLM-generated statements for a deliberation topic. For each trial, we randomly sample $K$ elections and combine their preference DAGs as the disjoint union $P=\bigsqcup_{i=1}^{K}P_i$, yielding a disconnected DAG over $4K$ uniquely identified candidates. Each question-response pair is embedded using \texttt{all-mpnet-base-v2} \citep{song2020mpnet}, yielding
$768$-dimensional representations.
\Cref{tab:real-500-main} reports mean violations (total, PO, and majority), total slack for $\mathsf{LP2}$ and $\mathsf{LP3}$, and optimization time. $\mathsf{LP3}$ achieves the lowest violation in all 3 categories, while $\mathsf{LP2}$ has the lowest optimization time and $\|\vepsilon\|_1$. Linear BTL exhibits substantially higher violation rates and considerably larger optimization time. Additional results (varying $m,d,n$) are in~\appref{appdx:experiments}.

\section{Discussion and future work}

In this paper, we investigated how to satisfy desirable social choice axioms in linear reward models for AI alignment. We proposed an implementable algorithm that learns a linear reward model, inducing a ranking with fewer violations of PO and PMC. Our algorithm accommodates both synthetic and real-world preference data, including partial preferences and Condorcet cycles, and outperforms the linear BTL model. These findings open several directions for future research: Can our approach be extended to implementable non-linear reward models, such as convex models? What structural properties of the underlying LPs can enable better reward modeling and help formalize, analyze, and enhance LLM behavior? What would be the most efficient approximation of LP to implement it efficiently for large $m$?


\newpage
\bibliography{references}
\bibliographystyle{plainnat}

\newpage
\appendix

\section{Deferred Proofs from Section 3}\label{appdx:proofs}

\lemFirstLPAxioms*
\begin{proof}
    Let $\gC=\{c_1, c_2, \dots c_m\}$ be the set of candidates and let $\pi$ be a preference profile over $\gC$. According to the preference profile, we get the sets $\gU$ and $\gM$, which contain the unanimous and majority pairs, respectively. Let $(\hat{\vtheta},\hat{\vepsilon})$ be any feasible solution of $\mathsf{LP1}$.
    
    \paragraph{Pareto Optimality.} Irrespective of the value of $\mathds{1}_{PMC}$, the set of constraints in $\mathsf{LP1}$ for the unanimous pairs $\gU$ is preserved. Therefore, for every Pareto-optimal pair $(c_i, c_j) \in \gU$, we have
    $
    \langle\hat{\vtheta},\vx_{c_i}-\vx_{c_j}\rangle
    +\hat{\varepsilon}_{c_i}-\hat{\varepsilon}_{c_j}
    \ge \eta,$
    which implies $r_{\hat{\vtheta},\hat{\vepsilon}}(c_i)> r_{\hat{\vtheta},\hat{\vepsilon}}(c_j)$. Thus, the induced reward function satisfies PO.

    \paragraph{Pairwise Majority Consistency.} From the definition, if there is a Condorcet cycle, then PMC is vacuously true. Otherwise, we have that $\mathds{1}_{PMC} = 1$. As $\gU \subseteq \gM$, it implies $\gU \cup \mathds{1}_{PMC}\gM \ = \ \gU \cup \gM \ = \ \gM$.  Since $(\hat{\vtheta},\hat{\vepsilon})$ is feasible, the constraint for each $(c_i,c_j)\in \gM$ is satisfied. That is,
    $\langle\hat{\vtheta},\vx_a - \vx_b\rangle
    +\hat{\epsilon}_{c_i}-\hat{\epsilon}_{c_j}
    \ge \eta$.
    Hence $r_{\hat{\vtheta},\hat{\vepsilon}}(c_i)> r_{\hat{\vtheta},\hat{\vepsilon}}(c_j)$
    for every majority edge, establishing PMC. 
        
    \paragraph{Failure of Condorcet Consistency.} For a 1 dimensional feature space, let $\gC = \{c_1, c_2, c_3, c_4\}$, $x_{c_1} = 0$, $x_{c_2} = 1$, $x_{c_3} = 2$, and $x_{c_4} = 1.5$. Let there be 3 voters =\{1,2,3\}, whose preferences are $c_4 \succ_1 c_1 \succ_1 c_2 \succ_1 c_3$, $c_4 \succ_2 c_2 \succ_2 c_3 \succ_2 c_1$ and $c_3 \succ_3 c_1 \succ_3 c_2 \succ_3 c_4$, respectively. Based on the preference profile, $c_4$ is preferred to $c_i$ by a majority of voters for each $i \in \{1,2,3\}$. Hence, $c_4$ is a Condorcet winner. Since there are no unanimous pairs in the preference profile, $\gU = \emptyset$. Also, $c_1$ beats $c_2$, $c_2$ beats $c_3$ and $c_3$ beats $c_1$ in majority; this creates a cycle rendering $\mathds{1}_{PMC} = 0$.

    Therefore, the $\mathsf{LP1}$ for this instance has no constraint for (PO+PMC). Then, the following is a solution for $\mathsf{LP1}$: $\hat{\vepsilon} = \vzero$ and $\hat{\theta} = 1$ or $\hat{\theta} = -1$. In either of the solutions, $c_4$ is not the winner in the ranking induced by the reward function $r_{\vtheta}$. Hence, it fails Condorcet consistency.  
\end{proof}

\thmSecondLPAxioms*

In order to prove the theorem, we first need the following lemma describing the property of the feedback arc set we delete.
\begin{lemma} \label{lemma:coppersmith_pareto}
The feedback arc set $\gF$ obtained using the algorithm in~\cite{coppersmithOrderingWeightedNumber2010} has no Pareto-optimal (unanimous) edge.
\end{lemma}
\begin{proof}
Let $(c_1,c_2)$ be a unanimous pair. Thus every voter prefers $c_1$ to $c_2$, i.e., $n_{c_1\succ c_2}(\pi)=n$. Consider any candidate $x\in C\setminus\{c_1,c_2\}$. If $c_2\succ_i x$ for some voter $i$, transitivity of voter $i$'s preference implies
$c_1\succ_i x$. Hence $n_{c_1\succ x}(\pi)\ge n_{c_2\succ x}(\pi)$,
which means every majority victory of $c_2$ is also a majority victory of $c_1$. Therefore,
$\operatorname{out-deg}(c_1)\ge
\operatorname{out-deg}(c_2)+1$ in $G_{\gM}$,
since $c_1$ defeats $c_2$ and every out neighbor of $c_2$.
Equivalently, $\operatorname{in-deg}(c_1) < \operatorname{in-deg}(c_2)$.
The CopperSmith algorithm orders vertices by non-decreasing in-degree. Hence $c_1$ always appears before $c_2$, making $(c_1,c_2)$ a forward edge. However, the feedback arc set the set of backward edges in the ordering and does not contain the edge $(c_1,c_2)$. Thus, the algorithm never adds a PO edge in $\gF$. If $\gU$ is the set containing the unanimous pairs, $\gF \cap \gU = \emptyset$.
\end{proof}

Now we are ready to prove the theorem.

\begin{proof}[Proof of \Cref{thm:primal_lp}]
Let $\gC=\{c_1, c_2, \dots c_m\}$ be the set of candidates and let $\pi$ be a preference profile over $\gC$. Let $G_\gM$ denote the majority graph corresponding to $\pi$.
If $G_\gM$ contains a cycle, then we first apply the CopperSmith algorithm \citep{coppersmithOrderingWeightedNumber2010}. We order the vertices in the non-increasing order of their in-degree. The feedback arc set $\gF$ is the set of backward arcs in the ordering. If $G_\gM$ is acyclic, then $\gF=\emptyset$. Let the set of pairs $\gP = E(G_\gM) \setminus \gF$ be the remaining edges of $G_{\gM}$ after removing the feedback arc set. There is a constraint for each pair of candidates in $\gP$ in $\mathsf{LP2}$. Let $(\hat{\vtheta},\hat{\vepsilon})$ be any feasible solution of $\mathsf{LP2}$. We prove each property separately.

\paragraph{Pareto Optimality.} If $\gF=\emptyset$, every Pareto-optimal edge already belongs to $\gP$.
If $\gF \neq \emptyset$, we show in \Cref{lemma:coppersmith_pareto}  that the CopperSmith algorithm never removes a Pareto-optimal edge. Hence, every Pareto-optimal pair remains in $\gP$. Therefore, for every Pareto-optimal pair $(c_i,c_j)$ it holds that
$
\langle\hat{\vtheta},\vx_{c_i}-\vx_{c_j}\rangle
+\hat{\varepsilon}_{c_i}-\hat{\varepsilon}_{c_j}
\ge \eta$,
which implies $r_{\hat{\vtheta},\hat{\vepsilon}}(c_i)>
r_{\hat{\vtheta},\hat{\vepsilon}}(c_j)$.
Thus, the induced reward function satisfies PO.

\paragraph{Pairwise Majority Consistency.} From the definition, if there is a Condorcet cycle, then PMC is vacuously true. Otherwise, we have that $G_\gM$ is acyclic, which implies $\gF = \emptyset$, it means that $\gP = \gM$. Since $(\hat{\vtheta},\hat{\vepsilon})$ is feasible, each constraint for $\gP$ is satisfied:
$\langle\hat{\vtheta},\vx_a-\vx_b\rangle
+\hat{\epsilon}_{c_i}-\hat{\epsilon}_{c_j}
\ge \eta$,
for each  $(c_i,c_j)\in \gP.$
Hence, $r_{\hat{\vtheta},\hat{\vepsilon}}(c_i)>
r_{\hat{\vtheta},\hat{\vepsilon}}(c_j)$
for every retained majority edge, establishing PMC.

\paragraph{Condorcet Consistency.}
Suppose a candidate $c^* \in \gC$ is the Condorcet winner in the preference profile $\pi$. Then, for every candidate $c\in C\setminus\{c^*\}$, we have 
$w_{c^*,c}(\pi)>\frac{1}{2}$, that is, $c^*$ defeats every other candidate by majority.
Hence, in the majority graph $G_\gM$, the vertex $c^*$ has an outgoing edge to every other candidate and no incoming edge. Therefore,
$\operatorname{in-deg}(c^*)=0$ and 
$\operatorname{in-deg}(c)\ge 1$
for each $c\neq c^*$.
The Coppersmith algorithm orders vertices by non-decreasing in-degree. Since $c^*$ is the unique vertex with minimum in-degree, it is placed first in the ordering. Consequently, there is no backward arc incident to $c^*$, and none of its outgoing edges are removed while constructing the feedback arc set $\gF$. Thus, for every $c\neq c^*$, the edge $(c^*,c)$ belongs to the constraint set $\gP$. Since $(\hat{\vtheta},\hat{\vepsilon})$ is feasible, we get
$\langle \hat{\vtheta},x_{c^*}-x_c\rangle
+\hat{\varepsilon}_{c^*}-\hat{\varepsilon}_c
\ge \eta$.
Therefore, $r_{\hat{\vtheta},\hat{\vepsilon}}(c^*)
>
r_{\hat{\vtheta},\hat{\vepsilon}}(c)$ for each $ c\neq c^*$.
Hence, $c^*$ is ranked first in the ranking induced by the reward function, proving Condorcet Consistency.

\paragraph{Majority Consistency.} If candidate $c^* \in \gC$ is ranked first by a majority of the voters, it defeats all other candidates by majority, and $c^*$ becomes the Condorcet winner. It is shown that Condorcet consistency is maintained for every feasible solution of $\mathsf{LP2}$. So, $c^*$ is ranked first by $r_{\hat{\vtheta}, \hat{\vepsilon}}$. By the definition of majority consistency, $r_{\hat{\vtheta}, \hat{\vepsilon}}$ also satisfies majority consistency.

Thus, it is proved that any feasible solution of $\mathsf{LP2}$ induces a reward function using \Cref{eq:relaxed-reward} that satisfies PO, PMC, Condorcet Consistency, and Majority Consistency.  
\end{proof}

\section{Deferred proofs from Section 4}\label{appdx:section4}

\thmThirdLPSlack*

\begin{proof}
For each $\vtheta\in\mathbb{R}^d$, define the minimum total slack satisfying margin constraints as
\[
c(\vtheta) := \min\Bigl\{\|\vepsilon\|_1 :\ \varepsilon_a - \varepsilon_b \ge \eta - \langle\vtheta, \vx_a - \vx_b\rangle\ \ \forall (a,b)\in \gP\Bigr\}.
\]
Let $(\vtheta^*_\lambda,\vepsilon^*_\lambda,\vu^*)$ be an optimal solution of $\mathsf{LP3}$. Then $\|\vepsilon^*_\lambda\|_1 = c(\vtheta^*_\lambda)$. To see this, suppose
$\|\vepsilon^*_\lambda\|_1 > c(\vtheta^*_\lambda)$, and let some other slack $\widehat\vepsilon$ attain
$c(\vtheta^*_\lambda)$. Then $(\vtheta^*_\lambda, \widehat\vepsilon, \vu^*)$ is a feasible point of $\mathsf{LP3}$. It leaves $\sum \vu^*_{ab}$ term unchanged, and thus has a strictly smaller objective, a contradiction.

Now we bound $c(\vtheta)$ for any
$\vtheta$ with $\|\vtheta\|_\infty \le 1$. For $(a,b)\in \gP$, H\"older's inequality gives
\begin{equation}\label{eqn:holderbound}
|\langle\vtheta, \vx_a - \vx_b\rangle| \le \|\vtheta\|_\infty \|\vx_a - \vx_b\|_1 \le \|\vx_a - \vx_b\|_1 \le  \Delta\sqrt{d} \implies \eta - \langle\vtheta, \vx_a - \vx_b\rangle \le \eta + \Delta\sqrt{d}.
\end{equation}
The second inequality follows from the fact that $\|\vtheta\|_\infty \le 1$. 
Since $\gP$ is acyclic, without loss of generality, we assume it admits a linear extension
$a_1 \succ \ldots \succ a_m$: every pair has the form $(a_i, a_j)$ with $i<j$. Setting 
$\varepsilon_{a_i} = (\eta+\Delta\sqrt{d})\big(\tfrac{m+1}{2} - i\big)$, we can satisfy every margin
constraint since
\[
\varepsilon_{a_i} - \varepsilon_{a_j} =(\eta+\Delta\sqrt{d})(j-i) \ge \eta+\Delta\sqrt{d}.
\]
Hence, its $\ell_1$-norm upper bounds $c(\vtheta)$, i.e., 
\[
\sup_{\|\vtheta\|_\infty\le 1} c(\vtheta) \le \|\vepsilon\|_1 = (\eta+\Delta \sqrt{d})\sum\nolimits_{i=1}^{m}\Big|\tfrac{m+1}{2}-i\Big|
\le (\eta+\Delta \sqrt{d})\Big\lfloor\tfrac{m^2}{4}\Big\rfloor.
\]
Since $\vtheta^*_\lambda$ lies in the box, it yields
$\|\vepsilon^*_\lambda\|_1 = c(\vtheta^*_\lambda) \le (\eta+\Delta\sqrt{d})\lfloor m^2/4\rfloor$.

Now, set $\vtheta = \vzero$, $\varepsilon_{a_i} = \eta\bigl(\tfrac{m+1}{2} - i\bigr)$ along the same linear extension, and $u_{ab} = \eta$. Then, in \Cref{eq:countViolations}, we get $\eta - \langle\vtheta, \vx_a - \vx_b\rangle \le \eta$. Moreover, every margin constraint reads $\varepsilon_{a_i} - \varepsilon_{a_j} = \eta(j-i) \ge \eta$, and all constraints hold. The objective value of $\mathsf{LP3}$ for this feasible point is $\eta\lfloor m^2/4\rfloor + \lambda\eta|\gP|$. Optimality, together with $\lambda\ge 0$ and $\vu^* \ge 0$, gives
\[
\|\vepsilon^*_\lambda\|_1 
\le \eta\bigl\lfloor m^2/4\bigr\rfloor + \lambda\bigl(\eta|\gP| -\sum\nolimits_{(a,b)\in \gP} u^*_{ab}\bigr) \le \eta\bigl\lfloor m^2/4\bigr\rfloor + \lambda\eta|\gP| .
\]
Combining both bounds proves the upper bound.

For the lower bound, take a unanimous profile in which every voter ranks
$a_1 \succ \cdots \succ a_m$. Then $G_\gM$ is the complete order and
acyclic, so $\gF = \emptyset$ and $\gP = \gU= \gM$. Now, fix $\vx_0 \in \mathbb{R}^d$,
pairwise distinct unit vectors $\vu_1, \dots, \vu_m$, and place the candidate features
at $\vx_{a_i} := \vx_0 + \gamma\,\vu_i$ with $\gamma := \eta / (4\sqrt{d})$.
Since $\|\vu\|_1 \le \sqrt{d}\,\|\vu\|_2 \le \sqrt{d}$, every $\vtheta$ in the box $\|\vtheta\|_\infty \le 1$ satisfies
\[
\big|\langle\vtheta,\, \vx_{a_i} - \vx_{a_j}\rangle\big|
\le \|\vtheta\|_\infty\,\|\vx_{a_i} - \vx_{a_j}\|_1
\le \gamma\big(\|\vu_i\|_1 + \|\vu_j\|_1\big)
\le 2\gamma\sqrt{d} = \tfrac{\eta}{2}.
\]
So, for any feasible slack $\vepsilon$, the margin constraint forces
$\varepsilon_{a_i} - \varepsilon_{a_j} \ge \eta/2$ for each $i < j$. To find the slack with minimum $\ell_1$-norm, pair $i$-th candidate with $(m+1-i)$-th for $i \le \lfloor m/2 \rfloor$ and see that $\|\vepsilon\|_1 = \sum_{i=1}^{\lfloor m/2 \rfloor}
\left(|\varepsilon_{a_i}| + |\varepsilon_{a_{m+1-i}}|\right)$.
The minimum is attained by
$\vepsilon^*_{a_i} = \tfrac{\eta}{2}\big(\tfrac{m+1}{2} - i\big)$, yielding 
\[
\|\vepsilon^*\|_1 \ge \sum\nolimits_{i=1}^{\lfloor m/2 \rfloor}
\left(\varepsilon^*_{a_i} - \varepsilon^*_{a_{m+1-i}}\right)\ge \tfrac{\eta}{2} \sum\nolimits_{i=1}^{\lfloor m/2 \rfloor} (m+1-2i) \ge \tfrac{\eta}{2} \Big\lfloor \tfrac{m^2}{4}\Big\rfloor,
\]
which completes the proof.
\end{proof}

\lemminKTdist*

\begin{proof} Recall that $r_\vtheta(c) = \langle\vtheta,\vx_c\rangle$ for each candidate $c \in \mathcal{C}$. Thus, $r_\vtheta$ violates the constraint for a pair $(a,b) \in \gP$ when $\langle\vtheta, \vx_a-\vx_b\rangle < \eta$. Due to the constraints in \Cref{eq:countViolations}, we have that if $\langle\vtheta, \vx_a-\vx_b\rangle < \eta$, then $u_{ab}$ is non-zero for each $(a,b) \in \gP$. Therefore, if $(a,b)$ is a violated pair, then $u_{ab}=1$  since $u_{ab} \in \{0,1\}$, and it is zero otherwise. That is, the variables $u_{ab}$  count the number of violated constants when we only consider $\vtheta$ to compute the reward function. For a large enough $\lambda$, the ILP minimizes the sum $\sum_{(a,b) \in \gP} u_{ab}$. Therefore, it minimizes the number of violations.
\end{proof}

\lemILPlambda*

\begin{proof}
    For a feasible solution $s = (\vtheta, \vepsilon, \vt, \vu)$ of $\mathsf{ILP}$, let $f(s) = \sum_{a \in \gC}t_a$ and $v(s) = \sum_{(a,b)\in \gP}u_{ab}$. Then we can write the objective function of $\mathsf{ILP}$ as
    $g(s, \lambda) = f(s) + \lambda v(s)$.
    
    Let $\lambda_1 < \lambda_2$. For a given preference profile $\pi$, the parameter $\eta$, feature vectors, and $\gP$, the $\mathsf{ILP}$ gives an optimal solution $s_1$ for $\lambda_1$ and an optimal solution $s_2$ for $\lambda_2$. 
    Since $\mathsf{ILP}$ is a minimization problem and $s_1$ is an optimal solution for $\lambda_1$, then we have the following,
    \begin{align}
        g(s_1, \lambda_1) \ & \leq \ g(s_2, \lambda_1) \nonumber\\
        f(s_1) + \lambda_1 v(s_1) \ & \leq \ f(s_2) + \lambda_1 v(s_2) \nonumber\\
        f(s_1) - f(s_2) \ & \leq \ \lambda_1 \big( v(s_2) - (s_1) \big) \label{eq:lambda1}
    \end{align}
    Similarly, since $s_2$ is an optimal solution for $\lambda_2$, then we have the following,
    \begin{align}
        g(s_2, \lambda_2) \ & \leq \ g(s_1, \lambda_2) \nonumber\\
        f(s_2) + \lambda_2 v(s_2) \ & \leq \ f(s_1) + \lambda_2 v(s_1) \nonumber\\
        f(s_2) - f(s_1) \ & \leq \ \lambda_2 \big( v(s_1) - (s_2) \big) \label{eq:lambda2}
    \end{align}
    Adding \Cref{eq:lambda1} and \Cref{eq:lambda2}, we get the following,
    \begin{align*}
        \lambda_2 \big( v(s_1) - v(s_2) \big) + \lambda_1 \big( v(s_2) - v(s_1) \big) & \geq 0 \\
        (\lambda_2 - \lambda_1) \big( v(s_1) - v(s_2) \big) & \geq 0 
    \end{align*}
    Since, $\lambda_1 < \lambda_2$, we have $v(s_1) \geq v(s_2)$, which means as $\lambda$ increases, the number of violations decreases monotonically. This proves the first part of the theorem.

    Recall that there is a violation for $(a, b)\in \gP$ when $\langle\vtheta, \vx_a-\vx_b\rangle < \eta$. Due to constraint~(\ref{eq:countViolations}), we have that $u_{ab} = 1$ for a violated pair $(a, b)$, and it is zero otherwise. Thus, $v(s)$ represents the number of violated pairs in $\gP$.
    Now, let $p = |\gP|$. Then, $v(s) \in \{0, 1, \ldots, p \}$, can take a finite number of values since $p$ is finite for a finite number of candidates. Let $v_{min}$ be the minimum number of violations possible for a problem instance and $s_{min}$ is a feasible solution attaining $v(s_{min}) = v_{min}$. Therefore, the  value of the objective function at $s_{min}$ for a given $\lambda$ is
    \begin{align*}
        g(s_{min}, \lambda) & = f(s_{min}) + \lambda v(s_{min}) \\
                            & = f(s_{min}) + \lambda v_{min}
    \end{align*}
    In \Cref{thm:upper-bound-LP3}, we give an upper bound in $f$ for the relaxed $\mathsf{LP3}$. Now, we want a bound on $f$ for $\mathsf{ILP}$. Thus, in the proof of \Cref{thm:upper-bound-LP3} we set $u_{ab} = 1$ instead of $u_{ab}=\eta$. Then, we get an upper bound on $f$ and $g$ as follows.
    \begin{align}
        f(s_{min}) \ & \leq \quad \min\{\eta  \lfloor \tfrac{m^2}{4} \rfloor +  \lambda p,  \ (\eta + \Delta\sqrt{d})\lfloor\tfrac{m^2}{4}\rfloor \} \nonumber\\
        f(s_{min}) & \leq^{(\mathrm{i})}  (\eta + \Delta\sqrt{d})\lfloor\tfrac{m^2}{4}\rfloor \nonumber\\
        f(s_{min}) + \lambda v_{min} & \leq \quad (\eta + \Delta\sqrt{d})\lfloor\tfrac{m^2}{4}\rfloor + \lambda v_{min} \nonumber\\
        g(s_{min}, \lambda) & \leq \quad (\eta + \Delta\sqrt{d})\lfloor\tfrac{m^2}{4}\rfloor + \lambda v_{min} \label{eq:lambda3}
    \end{align}
    The inequality $(\mathrm{i})$ follows from the fact that $\min \{a, b\} \leq b$.
    
    Next, we show that if we set $\lambda >  (\eta + \Delta\sqrt{d})\lfloor\tfrac{m^2}{4}\rfloor$, then any feasible solution with number of violations more than $v_{min}$ has a strictly larger value for the objective.
    Let $s'$ be a feasible solution of $\mathsf{ILP}$ such that $v(s') \geq v_{min} + 1$. If we set $\lambda >  (\eta + \Delta\sqrt{d})\lfloor\tfrac{m^2}{4}\rfloor$, we get the following as 
    \begin{align*}
        \lambda \quad & > \quad (\eta + \Delta\sqrt{d})\lfloor\tfrac{m^2}{4}\rfloor \\
        \lambda + \lambda v_{min} & > \quad  (\eta + \Delta\sqrt{d})\lfloor\tfrac{m^2}{4}\rfloor  + \lambda v_{min} \\
        \lambda (v_{min} + 1) & > \quad  (\eta + \Delta\sqrt{d})\lfloor\tfrac{m^2}{4}\rfloor  + \lambda v_{min} \\ 
        f(s') + \lambda v(s') & >^{(\mathrm{ii})}  (\eta + \Delta\sqrt{d})\lfloor\tfrac{m^2}{4}\rfloor  + \lambda v_{min} \\ 
        g(s', \lambda) & > \quad  (\eta + \Delta\sqrt{d})\lfloor\tfrac{m^2}{4}\rfloor  + \lambda v_{min} \\
        g(s', \lambda) & >^{(\mathrm{iii})} g(s_{min}, \lambda)
    \end{align*}

    The inequality $\mathrm{(ii)}$ follows from $f(s') + \lambda v(s') \geq \lambda (v_{min} + 1)$ as $f(s') \geq 0$. The inequality $\mathrm{(iii)}$ follows from \eqref{eq:lambda3}. 

    The value of the objective function for any feasible solution containing at least $v_{min} + 1$ violations is strictly larger than a feasible solution containing $v_{min}$ violations. So, the minimization problem will reach a solution with $v_{min}$ violations, and increasing $\lambda$ will not decrease the number of violations any more since $v$ is a monotonically decreasing function in $\lambda$ and $v_{min}$ is the minimum number of violations possible for the problem instance. Hence, $\lambda^* > (\eta + \Delta\sqrt{d})\lfloor\tfrac{m^2}{4}\rfloor$ is cut-off value.

\end{proof}

\lemThirdLPAxioms*
\begin{proof}
    Let $\gI = (\gX, \gP)$ be a problem instance for $\mathsf{LP3}$ where $\gX$ and $\gP$ are the feature matrix and the preference DAG, respectively. Let $\gS_{\vtheta, \vepsilon, \vt, \vu}(\gI, \mathsf{LP3})$ be the set of feasible solutions of $(\vtheta, \vepsilon, \vt, \vu)$ for instance $\gI$ running $\mathsf{LP3}$.
    
    All the constraints set of $\mathsf{LP2}$ are present in $\mathsf{LP3}$, along with two extra set of constraints present in $\mathsf{LP3}$ to decrease the number of violations, i.e., an non-integral, relaxed version of \eqref{eq:countViolations}. Therefore any feasible solution of $\mathsf{LP3}$ (removing $\vu$) will be a feasible solution of $\mathsf{LP2}$, $\gS_{\vtheta, \vepsilon, \vt}(\gI, \mathsf{LP3}) \subseteq \gS_{\vtheta, \vepsilon, \vt}(\gI, \mathsf{LP2})$.

    Let $(\hat{\vtheta}, \hat{\vepsilon}, \hat{\vt}, \hat{\vu})$ be a feasible solution of $\mathsf{LP3}$ on problem instance $\gI$, hence, $(\hat{\vtheta}, \hat{\vepsilon}, \hat{\vt})$ will be a feasible solution of $\mathsf{LP2}$. From \Cref{thm:primal_lp}, we know that $r_{\hat{\vtheta}, \hat{\vepsilon}}$ will satisfy PO, PMC, Condorcet Consistency, and majority consistency. Since $r_{\hat{\vtheta}, \hat{\vepsilon}}$ is also the induced reward function from a feasible solution of $\mathsf{LP3}$, every feasible solution of $\mathsf{LP3}$ induces a reward function that also satisfies PO, PMC, Condorcet Consistency, and majority consistency.
\end{proof}

\proprealizability*
\begin{proof}
Write $\kappa := \kappa(\pi,\gX)$. Realizability gives a $\vtheta$ satisfying
$\langle\vtheta, \vx_a - \vx_b\rangle > 0$ on every pair of $\gP$. Rescaling it
by $1/\|\vtheta\|_\infty$ puts it in the box, yielding $\kappa  > 0$. Fix $\eta \le \kappa$ and let $\widehat \vtheta$ attain the maximum in \Cref{eq:kappa}. Then
$\langle\widehat \vtheta, \vx_a - \vx_b\rangle \ge \kappa \ge \eta$ for every
$(a,b) \in \gP$, so the pair $(\widehat \vtheta, \bm{0})$ satisfies every margin
constraint, and the shortfalls
$h_{ab}(\widehat \vtheta) = \eta - \langle\widehat\vtheta, \vx_a - \vx_b\rangle \le 0$ allow
$u_{ab} = 0$. This feasible point has the objective value $0$, hence it is optimal and thus
$\vepsilon^*_\lambda = \bm{0}$.
Now, fix $\eta > \kappa$ and suppose $\vepsilon^*_\lambda = \bm{0}$.
Each margin constraint then reads
$\langle\vtheta^*_\lambda, \vx_a - \vx_b\rangle \ge \eta$. So $\vtheta^*_\lambda$ is feasible for \eqref{eq:kappa}
with value at least $\eta$, giving $\kappa \ge \eta$ -- a contradiction. Hence
$\vepsilon^*_\lambda \ne \bm{0}$.
\end{proof}

\propgordan*

\begin{proof}
 Gordan's theorem~\citep{vandenberghe2004} says that for any matrix $\mA$, either
(i) $\exists\, \vz$ with $\mA\vz > \vzero$ componentwise, or,
(ii) $\exists\, \vy \gneq \vzero$, with $\mA^{\!\top}\vy = \vzero$. Take $\mA := \big[(\vx_a - \vx_b)^{\!\top}\big]_{(a,b)\in \gP}$.

(i) Taking $\vz = \vtheta$, the condition $\mA\vtheta > \vzero$ is exactly $\langle\vtheta, \vx_a - \vx_b\rangle > 0$ on every constrained pair and gives realizability. It is
also equivalent to realizability at any prescribed margin: taking
$\vz = \lambda \vtheta$, where
$\lambda = \eta/\min_{(a,b)\in \gP}\langle\vtheta, \vx_a - \vx_b\rangle$, gives $\langle\vtheta, \vx_a - \vx_b\rangle \ge \eta$ 

(ii) Taking $\vy = \vmu$, with $\vmu \gneq \vzero$, gives
$\mA^{\!\top}\vmu = \sum_{(a,b)\in \gP}\mu_{ab}(\vx_a - \vx_b)$. Regrouping the sum by candidates gives $\mA^{\!\top}\vmu = \sum_{a\in \gC}\nu_a(\vmu)\vx_a = \vzero$. 

Exactly one alternative holds, giving the claim.
\end{proof}


\begin{lemma}[Balanced flow need $d+2$ candidates]\label{lem:d-plus-2}
Let $\vmu \gneq \vzero$ satisfy $\sum_{a\in \gC}\nu_a(\vmu)\vx_a = \vzero$. Set
$\gT := \{ a \in \gC : \nu_a(\vmu) \ne 0\}$.
Then $\{\vx_a\}_{a \in \gT}$ is affinely dependent, with coefficients $(\nu_a(\vmu))_{a \in \gT}$. Moreover, if the features are in the general position, then
$|\gT| \ge d+2$.
\end{lemma}

\begin{proof}
Taking $\zeta_a = \nu_a(\vmu)$ gives
$\sum_{a \in \gC}\zeta_a \vx_a = 0$, and $\sum_{a \in \gC}\zeta_a = 0$. Restricting both sums to $\gT$ changes nothing, as $\lambda_a = 0$ off $\gT$.
It remains to show $\vlambda \not\equiv \vzero$, so that the dependence is nontrivial.

Suppose $\bm{\zeta} \equiv \vzero$. Since $\gP$ is acyclic, let $\varphi_a$ be the position of $a$ in a reversed topological order, so
$\varphi_a - \varphi_b \ge 1$ on every edge $(a,b)\in \gP$. Interchanging sums over candidates and edges, we get
\[
0 = \sum_{a \in \gC}\zeta_a \varphi_a
= \sum_{(a,b) \in \gP}\mu_{ab}(\varphi_a - \varphi_b)
\ge \sum_{(a,b) \in P}\mu_{ab} \ge 0.
\]
Hence $\vmu \equiv \vzero$, a contradiction. So $\bm{\zeta} \not\equiv \vzero$,
$\gT \ne \emptyset$, and $\{\vx_a\}_{a \in \gT}$ is affinely dependent. Under the general
position, this forces $|\gT| \ge d+2$ since every affinely dependent subset $\gT$ has at least $d+2$ elements.
\end{proof}

\corraffineindep*
\begin{proof}[Proof of \Cref{prop:affine-indep}]
By Lemma~\ref{lem:d-plus-2}, a non-zero feature balanced flow would require
$|\gT| \ge d+2 > m$ candidates. Hence, no such flow exists,
and Proposition~\ref{prop:gordan} gives linear realizability.
\end{proof}

\section{Additional Experimental Results} \label{appdx:experiments}

We conduct additional experiments to examine the robustness of our
observations under alternative synthetic and real-world settings. Unless otherwise stated, all experiments use the same evaluation metrics and optimization procedures as described in \Cref{sec:experiments}. 

\paragraph{Real world dataset description.} We perform our experiments on the \emph{Habermas Machine Deliberation} dataset, particularly PrefLib dataset 00070 \citep{tesslerPreflibDataset2024}. The dataset contains 2,710 separate election files. Each mini-jury consists of five participants who provide rankings over four statements generated by an LLM for a specific deliberation topic; the original data were collected by Google DeepMind and converted to PrefLib format.\footnote{\url{https://preflib.github.io/PrefLib-Jekyll/dataset/00070}} The resulting files therefore provide a collection of small, independently defined preference profiles. 

\paragraph{Experimental setup.} Python 3.13 \citep{python3} is used for running the experiments in this paper. The linear programs $\mathsf{LP2}$ and $\mathsf{LP3}$ are solved using the CVXPY \citep{diamond2016cvxpy,agrawal2018rewritingCVXPY} package. Linear BTL model function is trained using a non-linear constrained optimization algorithm, Sequential Least Squares Programming (SLSQP) with a maximum iteration count of $100$ and a tolerance of $10^ {-6}$, on a constraint of $\| \vtheta \|_\infty \leq 1$. SLSQP is solved using the SciPy package \citep{2020SciPy-NMeth}. For the embedding of the real-world data questions and responses, we utilize the \texttt{all-mpnet-base-v2} variant of the MPNet model \citep{song2020mpnet}. The model is accessed online from the \texttt{huggingface}\footnote{\url{https://huggingface.co/sentence-transformers/all-mpnet-base-v2}} library.
In the real-world data, we sample $K$ questions and get $K$ preference DAGs. If we denote $\gP_i$ as the edge set of the $i^{th}$ preference DAG, for the loss calculation of linear BTL, we used the following log-likelihood-based loss function,
$$
    \gL_{\text{Linear-BTL}}(\vtheta) := - \sum_{i = 1}^{K} \sum_{(a, b) \in \gP_{i}} n_{a \succ b} \log(\sigma(\langle \vtheta, \vx_a - \vx_b \rangle)) + n_{b \succ a} \log(\sigma(\langle \vtheta, \vx_b - \vx_a \rangle))
$$
We combine the preference graphs to produce the larger graph, which consists only of those smaller preference graphs as disconnected components. So, this loss function is justified.

\paragraph{Additional synthetic experiments.}
We vary the synthetic problem parameters, including the number of
alternatives, the number of voters, and the feature dimension, while keeping the remaining components of the experimental pipeline fixed. Results are reported in \Cref{fig:mallows-dispersion_3_10_5,fig:mallows-dispersion_3_15_5}.

\paragraph{Additional real-world experiments.} We use the same experimental setup as used in \Cref{sec:experiments}. We perform several experiments for $100$ and $500$ questions with varying values of $d < 768$. PCA is used to reduce the feature dimension to the desired dimension. The findings are reported in \Cref{tab:real_data_results_100,tab:real_data_results_500}.

\begin{figure}[t]
    \centering
    \includegraphics[width=\linewidth]{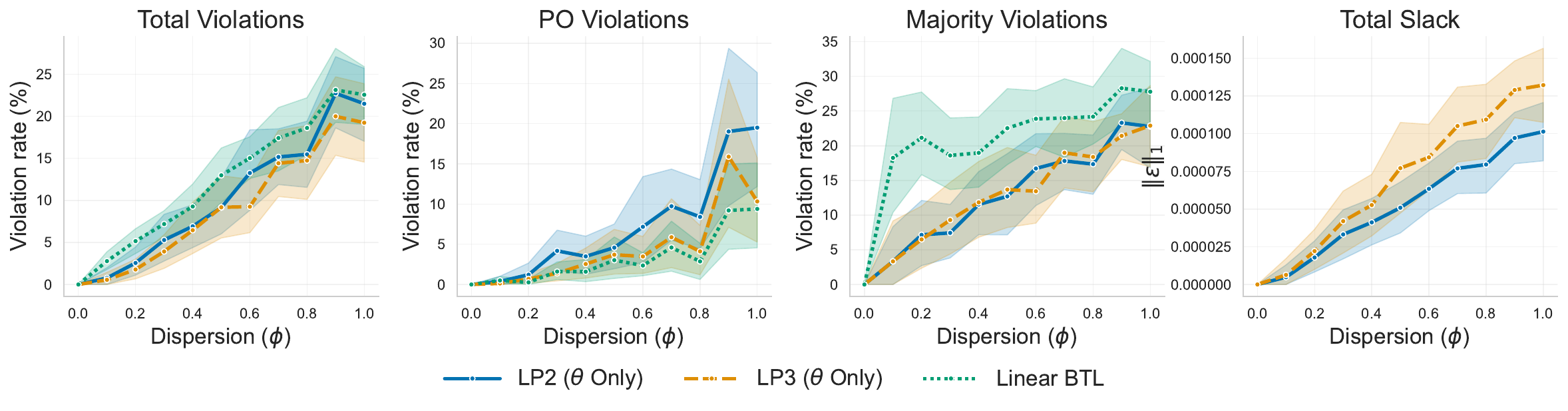}
    \caption{
        Axiom violations and candidate-level slack as a function of normalized Mallows dispersion $\phi$.
        Curves show mean violation rates across trials, with shaded regions denoting 95\% confidence intervals. Panels report (a) Total violations, (b) PO violations, (c) Majority violations,
        and (d) The $\|\vepsilon\|_1$
        for the $\mathsf{LP2}$ and $\mathsf{LP3}$.
        Experiments use $3$ voters, $10$ candidates, and
        $5$-dimensional feature embeddings.}
    \label{fig:mallows-dispersion_3_10_5}
\end{figure}

\begin{figure}[t]
    \centering
    \includegraphics[width=\linewidth]{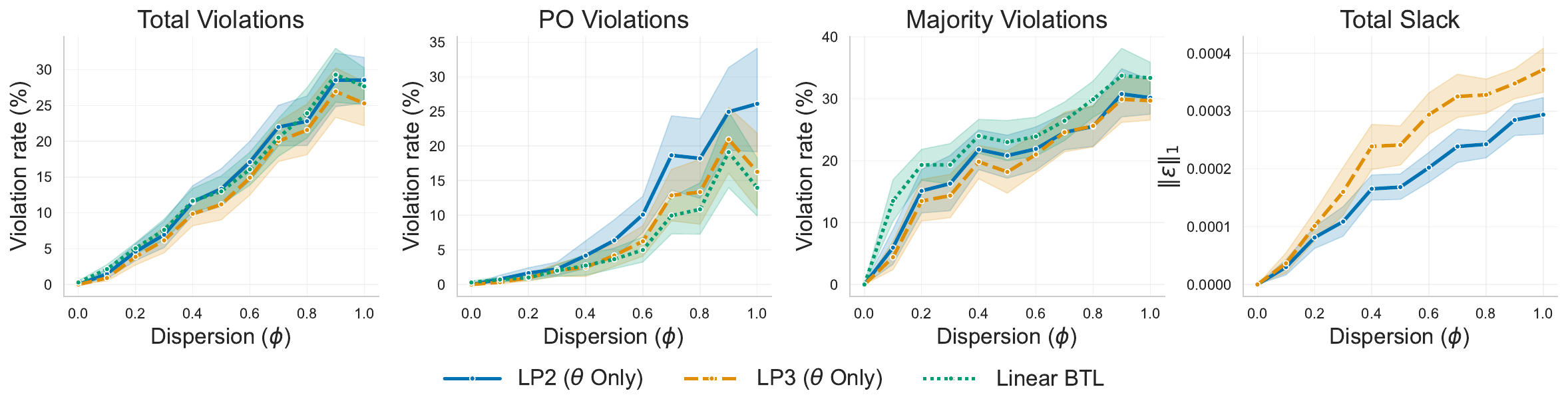}
    \caption{
        Axiom violations and candidate-level slack as a function of normalized Mallows dispersion $\phi$.
        Curves show mean violation rates across trials, with shaded regions denoting 95\% confidence intervals. Panels report (a) Total violations, (b) PO violations, (c) Majority violations,
        and (d) The $\|\vepsilon\|_1$
        for the $\mathsf{LP2}$ and $\mathsf{LP3}$. Experiments use $3$ voters, $15$ candidates, and $5$-dimensional feature embeddings.}
    \label{fig:mallows-dispersion_3_15_5}
\end{figure}

\begin{table}[t]
    \centering
    \small
    \caption{
        Performance of $\mathsf{LP2}$, $\mathsf{LP3}$, and Linear BTL on the real-world benchmark
        of 100 questions comprising 400 candidates, evaluated using
        PCA-reduced feature representations of dimensions
        $d \in \{25,75,100\}$. Each entry reports the mean
        $\pm$ standard deviation across the 100 questions. $\eta = 10^{-5}$ for $\mathsf{LP2}$ and $\mathsf{LP3}$.
    }
    \label{tab:real_data_results_100}

    \begin{tabular}{lcccccc}
        \toprule
        Model
        & Total Viol.(\%)
        & PO Viol.(\%)
        & Majority Viol.(\%)
        & $\|\vepsilon\|_1 (\times 10^{-3})$
        & Time (s) \\
        
        \midrule
        \multicolumn{6}{c}{\textbf{(a) Feature dimension $d=25$}} \\
        \midrule
        
        $\mathsf{LP2}$
        & $28.39 \pm 2.04$
        & $15.60 \pm 3.55$
        & $\textbf{32.48} \pm \textbf{2.21}$
        & $2.7 \pm 0.155$
        & $\textbf{1.5}$ \\

        $\mathsf{LP3}$
        & $\textbf{27.95} \pm \textbf{2.31}$
        & $\textbf{13.69}\pm\textbf{3.41}$
        & $32.52\pm2.43$
        & $2.87 \pm 0.155$
        & $2.95$ \\

        Linear BTL
        & $31.30\pm2.54$
        & $14.93\pm2.99$
        & $36.53\pm2.66$
        & --
        & $23.12$ \\

        \midrule
        \multicolumn{6}{c}{\textbf{(b) Feature dimension $d=50$}} \\
        \midrule

        $\mathsf{LP2}$
        & $23.97 \pm 2.20$
        & $11.47 \pm 2.23$
        & $27.98 \pm 2.48$
        & $2.41 \pm 0.136$
        & $\textbf{1.66}$ \\

        $\mathsf{LP3}$
        & $\textbf{22.76} \pm \textbf{2.31}$
        & $\textbf{9.97} \pm \textbf{2.16}$
        & $\textbf{26.84} \pm \textbf{2.80}$
        & $2.55 \pm 0.144$
        & $3.23$ \\

        Linear BTL
        & $27.96 \pm 1.74$
        & $11.13 \pm 2.48$
        & $33.36 \pm 2.02$
        & -- 
        & $45.76$ \\

        \midrule
        \multicolumn{6}{c}{\textbf{(c) Feature dimension $d=75$}} \\[3pt]
        \midrule

        $\mathsf{LP2}$
        & $19.85 \pm 1.75$
        & $9.99 \pm 2.26$
        & $22.96 \pm 2.25$
        & $2.07 \pm 0.135$
        & $\textbf{1.73}$ \\
        
        $\mathsf{LP3}$
        & $\textbf{18.64} \pm \textbf{1.78}$
        & $\textbf{7.81} \pm \textbf{2.34}$
        & $\textbf{22.07} \pm \textbf{2.07}$
        & $2.24 \pm 0.149$
        & $3.49$ \\

        Linear BTL
        & $25.18 \pm 1.99$
        & $7.92 \pm 1.97$
        & $30.68 \pm 2.32$
        & -- 
        & $66.19$ \\

        \midrule
        \multicolumn{6}{c}{\textbf{(d) Feature dimension $d=100$}} \\
        \midrule

        $\mathsf{LP2}$
        & $15.86 \pm 1.68$
        & $7.58 \pm 2.02$
        & $18.50 \pm 2.09$
        & $1.7 \pm 0.148$
        & $\textbf{1.86}$ \\

        $\mathsf{LP3}$       
        & $\textbf{14.51} \pm \textbf{1.70}$
        & $5.91 \pm 2.19$
        & $\textbf{17.26} \pm \textbf{2.11}$
        & $1.89 \pm 0.164$
        & $4.06$ \\

        Linear BTL
        & $23.37 \pm 1.53$
        & $\textbf{5.70} \pm \textbf{1.99}$
        & $29.03 \pm 1.58$
        & -- 
        & $89.28$ \\

        \bottomrule
    \end{tabular}
\end{table}

\begin{table}[t]
    \centering
    \small
    \caption{
        Performance of $\mathsf{LP2}$, $\mathsf{LP3}$, and Linear BTL on the real-world benchmark
        of 500 questions comprising 2000 candidates, evaluated using
        PCA-reduced feature representations of dimensions
        $d \in \{150,200,250,500\}$. Each entry reports the mean
        $\pm$ standard deviation across the 500 questions. $\eta = 10^{-5}$ for $\mathsf{LP2}$ and $\mathsf{LP3}$. 
    }
    \label{tab:real_data_results_500}

    \begin{tabular}{lcccccc}
        \toprule
        & Total Viol.(\%)
        & PO Viol.(\%)
        & Majority Viol.(\%)
        & $\|\vepsilon\|_1$
        & Time (s) \\

        \midrule
        \multicolumn{6}{c}{\textbf{(a) Feature dimension $d=150$}} \\
        \midrule
        $\mathsf{LP2}$
        & $24.54 \pm 1.18$
        & $10.90 \pm 1.38$
        & $29.02 \pm 1.21$
        & $1.26*10^{-2}\pm 3.57*10^{-4}$
        & $\textbf{38.73}$ \\

        $\mathsf{LP3}$
        & $\textbf{23.48} \pm \textbf{0.98}$
        & $9.26\pm0.99$
        & $\textbf{28.17}\pm\textbf{1.26}$
        & $1.31* 10^{-2} \pm 3.48*10^{-4}$
        & $109.24$ \\

        Linear BTL
        & $25.71\pm1.16$
        & $\textbf{8.34}\pm\textbf{0.91}$
        & $31.40\pm1.24$
        & --
        & $863.24$ \\

        \midrule
        \multicolumn{6}{c}{\textbf{(b) Feature dimension $d=200$}} \\
        \midrule

        $\mathsf{LP2}$
        & $23.26 \pm 1.27$
        & $9.93 \pm 1.44$
        & $27.64 \pm 1.27$
        & $1.19 * 10^{-2} \pm 3.92 * 10^{-4}$
        & $\textbf{48.1}$ \\

        $\mathsf{LP3}$
        & $\textbf{22.12} \pm \textbf{1.11}$
        & $8.44 \pm 1.31$
        & $\textbf{26.61} \pm \textbf{1.19}$
        & $1.24 * 10^{-2} \pm 4.16 * 10^{-4}$
        & $121.08$ \\

        Linear BTL
        & $24.40 \pm 1.24$
        & $\textbf{6.70} \pm \textbf{0.86}$
        & $30.20 \pm 1.29$
        & -- 
        & $1093.21$ \\

        \midrule
        \multicolumn{6}{c}{\textbf{(c) Feature dimension $d=250$}} \\
        \midrule
        
        $\mathsf{LP2}$
        & $21.69 \pm 0.73$
        & $9.51 \pm 1.82$
        & $25.69 \pm 0.92$
        & $1.12 * 10 ^{-2} \pm 3.37*10^{-4}$
        & $\textbf{60.96}$ \\

        $\mathsf{LP2}$
        & $\textbf{20.60} \pm \textbf{0.82}$
        & $8.04 \pm 1.37$
        & $\textbf{24.73} \pm \textbf{1.13}$
        & $1.19*10^{-2} \pm 3.37*10^{-4}$
        & $151.94$ \\ 

        Linear BTL
        & $23.91 \pm 0.99$
        & $\textbf{6.29} \pm \textbf{0.53}$
        & $29.69 \pm 1.18$
        & -- 
        & $1271.33$ \\

        \midrule
        \multicolumn{6}{c}{\textbf{(d) Feature dimension $d=500$}} \\
        \midrule
        
        $\mathsf{LP2}$
        & $14.57 \pm 0.73$
        & $6.74 \pm 0.75$
        & $17.24 \pm 0.79$
        & $7.87*10^{-3} \pm 2.62*10^{-4}$
        & $\textbf{275.71}$ \\

        $\mathsf{LP3}$
        & $\textbf{13.31} \pm \textbf{0.91}$
        & $5.07 \pm 0.95$
        & $\textbf{16.12} \pm \textbf{0.94}$
        & $8.83*10^{-3} \pm 2.81*10^{-4}$
        & $674.41$ \\

        Linear BTL
        & $20.50 \pm 0.86$
        & $\textbf{3.34} \pm \textbf{0.58}$
        & $26.35 \pm 0.98$
        & -- 
        & $2368.84$ \\
        
        \bottomrule
    \end{tabular}
\end{table}

\end{document}